\documentclass{article}
\usepackage[utf8]{inputenc}

\RequirePackage[OT1]{fontenc}
\RequirePackage{amsthm,amsmath}
\RequirePackage{natbib,enumerate}
\RequirePackage[citecolor=blue,urlcolor=blue]{hyperref}
\usepackage[dvipsnames,svgnames]{xcolor}
\usepackage{bbold}
\usepackage[normalem]{ulem}
\usepackage{tikz}
\usepackage{pgfplots}
\usepackage{booktabs}
\usepackage[ruled,vlined]{algorithm2e}
\usepackage{thmtools} 
\usepackage{subcaption}
\usepackage{float}
\usepackage{subcaption}
\usepackage{caption}
\usepackage{enumitem}
\usepackage{thmtools, thm-restate}
\newtheorem{example}{Example} 
\usepackage{tikz}
\usetikzlibrary{arrows.meta, positioning, matrix, fit, backgrounds, calc}

\usepackage[format=plain,labelfont=bf,textfont=it, font=small]{caption}

\usepackage{graphicx}
\usepackage{amsfonts}
\usepackage{wrapfig}
\usepackage{comment}
\usepackage{tabularx}
\newcommand{\Prob}{\ensuremath{{\mathds P}}} 

\newcommand{\R}{\ensuremath{{\mathds R}}}
\newcommand{\E}{\ensuremath{{\mathds E}}}
\newcommand{\X}{\ensuremath{{\mathcal X}}}

\usepackage{lscape}

\newcommand{\Ind}{\ensuremath{{\mathds{1}}}} 
\usepackage{dsfont}
\usepackage{caption}
\usepackage{float}

\DeclareMathOperator*{\argmax}{arg\,max}

\newcommand{\distas}[1]{\mathbin{\overset{#1}{\kern\z\sim}}}%
\newsavebox{\mybox}\newsavebox{\mysim}
\newcommand{\distras}[1]{%
  \savebox{\mybox}{\hbox{\kern3pt$\scriptstyle#1$\kern3pt}}%
  \savebox{\mysim}{\hbox{$\sim$}}%
  \mathbin{\overset{#1}{\kern\z@\resizebox{\wd\mybox}{\ht\mysim}{$\sim$}}}%
}

\usepackage{varioref}
\usepackage{cleveref}

\usepackage[textwidth=3cm]{todonotes}

\newcommand{\indep}{\raisebox{0.05em}{\rotatebox[origin=c]{90}{$\models$}}}

\numberwithin{equation}{section}
\theoremstyle{plain}
\newtheorem{proposition}{Proposition}[section]
\newtheorem{corollary}{Corollary}[section]
\newtheorem{lemma}{Lemma}[section]
\newtheorem{definition}{Definition}[section]

\newtheorem{assumption}{Assumption}[section]

\title{How to rank imputation methods?}
\date{}
\author{Jeffrey Näf$^1$, Krystyna Grzesiak$^2$, Erwan Scornet$^3$ \\
$^1$Research Institute for Statistics and Information Science, University of Geneva\\
$^2$Faculty of Mathematics and Computer Science, University of Wroc\l{}aw\\
        $ ^3$Sorbonne Universite and Universite Paris Cite, CNRS,\\
        Laboratoire de Probabilites, Statistique et Modelisation, F-75005 Paris
}

\begin{document}

\maketitle


\begin{abstract}
Imputation is an attractive tool for dealing with the widespread issue of missing values. Consequently, studying and developing imputation methods has been an active field of research over the last decade. Faced with an imputation task and a large number of methods, how does one find the most suitable imputation? Although model selection in different contexts, such as prediction, has been well studied, this question appears not to have received much attention. In this paper, we follow the concept of Imputation Scores (I-Scores) and develop a new, reliable, and easy-to-implement score to rank missing value imputations for a given data set without access to the complete data. In practice, this is usually done by artificially masking observations to compare imputed to observed values using measures such as the Root Mean Squared Error (RMSE). We discuss how this approach of additionally masking observations can be misleading if not done carefully and that it is generally not valid under MAR. We then identify a new missingness assumption and develop a score that combines a sensible masking of observations with proper scoring rules. As such the ranking is geared towards the imputation that best replicates the distribution of the data, allowing to find imputations that are suitable for a range of downstream tasks. We show the propriety of the score and discuss an estimation algorithm involving energy scores. Finally, we show the efficacy of the new score in simulated data examples, as well as a downstream task.
\end{abstract}






\section{Introduction}
Missing values are ubiquitous in data analysis, as values may not be recorded for a variety of reasons. In such cases, one observes masked points according to various patterns of missingness and in addition to the data distribution, the potential missingness mechanism has to be considered. Following the seminal work of \citet{Rubin_Inferenceandmissing}, one usually differentiates between missing completely at random (MCAR), wherein the missingness is completely independent of the data, missing at random (MAR), when the probability of an entry being missing depends only on observed data and missing not at random (MNAR), when the mechanism is neither MCAR nor MAR.

A critical and widely used tool for dealing with missing values is imputation, whereby missing observations are replaced by ``reasonable'' values. As such, developing imputation methods has been an active area of research over the past decades and a wealth of methods has emerged. Most of these methods can be divided into \emph{joint modeling methods} that impute the data using one (implicit or explicit) model and the \emph{fully conditional specification} (FCS) where a different model for each dimension is trained \citep{FCS_Van_Buuren2007, VANBUUREN2018}. Joint modeling approaches may be based on parametric distributions \citep{schafer1997analysis}, or specified in a nonparametric way using tools such as generative adversarial networks (GANs) (\citet{GAIN, directcompetitor1, directcompetitor2}) or variational autoencoders (VAEs)(\citet{MIWAE,VAE1,VAE2, VAE3}). In contrast, in the FCS approach, the imputation is done one variable at a time, based on conditional distributions \citep[see, e.g.,][]{sequentialapproach0noaccess, sequentialapproach1, sequentialapproach2, greatoverview}. The most prominent example of FCS is the multiple imputation by chained equations (MICE) methodology \citep{mice}. 

In order to impute an observation in a certain pattern, a (conditional) distribution must be learned from different patterns. For instance, consider a survey on income and age, where income might be missing and age is always observed. In this case, to impute income, one needs to estimate the conditional distribution of income given age in the pattern in which both are observed to impute in the pattern in which income is missing. A crucial difficulty in this task is that for any missingness mechanism except MCAR, covariate distribution shifts naturally arise when changing from one pattern to another. For example, older participants in the survey might be less likely to reveal their income. Despite this being a classical MAR scenario, there is a shift in distribution from the pattern with observed income to the pattern with missing income. While it is generally acknowledged that FCS can impute MAR data \citep{MICE_Results, directcompetitor0, näf2024goodimputationmarmissingness}, nonparametric joint modeling approaches are usually designed for imputation under MCAR \citep{directcompetitor1, hyperimpute}.


The wealth of different imputation methods and their different properties raises a natural question: how can one choose the best imputation in general? In answering this question, we do not want to focus on a specific downstream task, such as prediction of a target variable, but instead desire to find the imputation method that best replicates the underlying distribution and thus performs well over a range of downstream tasks.  The main approach to identify well-performing general-purpose imputation methods has been to compare the quality of different imputation procedures via pointwise metrics (see below) in benchmark studies with artificially generated missing data. Although these are important considerations, this approach has serious limitations. First, most benchmarking papers use RMSE or MAE between imputed and true data points as a measure of imputation quality; see, e.g., \citet{benchmark_abidin2018performance, benchmark_metabolomics1, awesomebenchmarkingpaper, benchmark_xu2020, benchmark_wang2022deep, benchmark_ge2023simulation, benchmark_alam2023investigation, benchmark_pereira2024imputation, benchmark_deforth2024performance, chilimoniuk2024imputomics, benchmark_joel2024performance} among others. However, measures like RMSE favor methods that impute conditional means, instead of draws from the conditional distribution. Hence, using RMSE as a validation criterion tends to favor imputations that artificially strengthen the dependence between variables, thus leading to severe biases in parameter estimation and uncertainty quantification, as previously pointed out \citep{VANBUUREN2018, RFimputationpaper, ImputationScores,Naturepaper, näf2024goodimputationmarmissingness}. This led to a widespread recommendation of methods such as missForest \citep{stekhoven_missoforest}, which is very successful in imputing conditional means, while  distorting the data distribution. Second, which missingness mechanism is prevalent in real data is generally unknown, and benchmark studies tend to focus on the same simple mechanisms. For instance, even MCAR can lead to missingness with complex dependencies between dimensions. However, benchmarking papers considering MCAR usually generate missingness components as independent Bernoulli with the same parameter. 
Thus, even if some MCAR mechanisms would be realistic, the MCAR mechanisms considered in benchmarking studies appear unlikely to reflect real-life (MCAR) missingness mechanisms. This also holds more generally for M(N)AR mechanisms. 


As such, faced with an imputation task, it would be desirable to choose the imputation method that most closely replicates the data distribution directly for a given dataset. While there are a few papers that go beyond RMSE and MAE in comparing imputation methods when the complete data is available (e.g., in a benchmarking study with artificial missingness), such as \citet{ross,OTimputation, Naturepaper, näf2024goodimputationmarmissingness}, none of these consider the case where the underlying complete data is not available. The first important contribution to solving this problem was made in \cite{ImputationScores}, who define the concept of ``proper'' imputation scores (I-Scores) to rank imputations. Adapting the concept of proper scores in prediction \citep{gneiting}, an I-Score is proper (or meets propriety) if it provably ranks highest an imputation from the true data generating process. \cite{ImputationScores} introduced the concept of I-Scores and developed the DR-I-Score that is shown to be proper under a condition weaker than MCAR, which we refer to as conditional independence missing at random (\ref{CIMAR}, see Section \ref{Sec_Background} for details). The score works by comparing fully observed observations with imputed ones using an estimate of the Kullback-Leibler (KL) divergence obtained with a (Random Forest) classifier. Although comparing observations from different patterns is generally not valid outside of MCAR, the form of the KL divergence allowed to obtain propriety under \ref{CIMAR}.


In this paper, we build on the work of \cite{ImputationScores}, and develop a new I-Score, the ``energy-I-Score''. The score works by comparing (univariate) observed values with values drawn from the imputation distribution (i.e., imputed values) using the energy score \citep{gneiting}. We first discuss the difficulty of scoring under MAR, motivating a missingness assumption under which the new score is shown to be proper. We then discuss an estimation method that uses multiple imputation as a natural way to generate samples to compare with observed points using the energy score. Compared to the DR-I-Score, the method is closer to the concept of prediction scores \citep{gneiting} and naturally scores the ability of an imputation to generate samples. Crucially, it also does not rely on the performance of an additional classifier. Indeed, a poor performance of the classifier can lead to incorrect rankings, as we illustrate in Section \ref{Sec_Simulation}. In contrast, our new score lets the imputation method ``speak for itself'', by scoring the multiple imputations generated for observed columns directly using the energy score. Moreover, the method only scores univariate observation, which is easier to achieve than the comparison of multivariate samples. As an added benefit, methods that cannot produce multiple imputations, such as methods that impute conditional means, are severely punished, which is arguably a desirable outcome. 


The remainder of the article is organized as follows. After introducing our notation in Section \ref{notationsec}, Section \ref{Sec_Background} provides crucial background to develop the new score, including an overview of different MAR mechanism and discusses the difficulty of scoring under MAR. In Section \ref{Sec_Scoring}, we  introduce our new score, prove that it is proper and  presents a way to estimate it in practice. Sections \ref{Sec_Simulation} and \ref{Sec_Empirical} demonstrate the capabilities of the new score with simulated examples and with an empirical application using downstream tasks. The code for replicating the experiments and using the new scoring methodology can be found in \url{https://github.com/KrystynaGrzesiak/ImputationScore}.

\subsection{Notation}\label{notationsec}

Let $(\Omega, \mathcal{A}, \Prob)$ be the underlying probability space on which all random elements are defined. We assume to observe $n$ i.i.d. copies of $(X,M)$, where $M$ takes values in $\{0,1\}^d$ and $X$ is a masked version of a $d$-dimensional random vector $X^*$:
\begin{align*}
    X_j=\begin{cases}
       X_j^*, &\text{ if } M_j=0\\
       \texttt{NA}, &\text{ if } M_j=1.
    \end{cases}
\end{align*}
For example, the observation $(X_1, \texttt{NA},X_3)$ corresponds to the pattern $M=(0,1,0)$. The support of $X^*$ is denoted as $\mathcal{X} \subset \R^d$ and the one of $M$ as $\mathcal{M} \subset \{0,1\}^d$. To define assumptions on the missingness mechanism, we use a notation along the lines of \cite{whatismeant}. For a realization $m$ of the missingness random vector $M$, we let 
\begin{align*}
 o(X,m)& :=(X_j)_{j \in \{1,\ldots,d\}:m_j=0} \quad \textrm{be the observed part of $X$ according to $m$}, \\
 o^c(X,m)&:=(X_j)_{j \in \{1,\ldots,d\}:m_j=1} \quad \textrm{be the missing part of $X$ according to $m$}.
\end{align*}
Throughout, we take $\mathcal{P}$ to be a collection of probability measures on $\R^d \times \{0,1\}^d$ with finite first moment, dominated by some $\sigma$-finite measure $\mu$. Let $P^* \in \mathcal{P}$ be the (unobserved) joint distribution of $(X^*,M)$ with density $p^*(x,m)$, and $P_{X|M}^*$ be the conditional distribution of $X^*\mid M$, with density $p^*(x \mid M=m)$. We also define the support of the observed variable for a pattern $m$, $\X_{\mid m}=\{x \in \X: p^*(o(x,m) \mid M=m) > 0\}$. Finally, $P_{M}$ is the marginal distribution of the mask variable $M$. Given its discrete nature, we introduce the probability mass function $\Prob(M=m)$ for a given pattern $m\in\mathcal{M}$ so that for every measurable set $A\subset \{0,1\}^d$, $P_M[A]=\sum_{m \in A} \Prob(M=m)$. Thus it holds that $p^*(x,m)=p(x^* \mid M=m) \Prob(M=m)$. We take $P$ to be the distribution of $X$ with density $p$. That is $P$ is an extended measure, for a vector $x$ with potential missing values, $p(x)=p^*(o(x,m)\mid M=m) \Prob(M=m)$. We refer to \cite{MNARcontamination} for details.




For a given variable $j$, let $L_{j}=\{m \in \mathcal{M}: m_j=0  \},$ be the set of patterns in which the $j$th variable is observed, such that $\Prob(M_j=1)=\sum_{m \in L_j} \Prob(M=m)$. 
%
We first refine the definition of an imputation distribution in \citet{ImputationScores}: We define $\mathcal{H}_{P} \subset \mathcal{P}$ to be the set of imputation distributions compatible with $P$, that is\footnote{We note that while $h$ and $p$ are densities on $\R^d$, notation is slightly abused by using expressions such as $h(o(x,m)|M=m)$ and $p(o(x,m)|M=m)$, which are densities on $\R^{|\{j:m_j=0 \}|}$.} 
\begin{align}\label{imputationdistset}
      \mathcal{H}_{P}:= \{H \in \mathcal{P}: &\ H \text{ admits density } h \text{ and } h(x, m)=h(x \mid M=m) \Prob(M=m), \nonumber \\
      &\ h(o(x,m)|M=m)=p(o(x,m)|M=m),\text{ for all } m \in \mathcal{M} \}.
\end{align}
In the above definition, both $P$/$p$ and $\Prob_M$ are fixed, but the imputations $h(o^c(x,m)|o(x,m), M=m)$ vary.
For two distributions $P_1,P_2$ on $\R^d$ with densities $p_1$, $p_2$, we define the Kullback-Leibler divergence (KL divergence) as 
     \begin{align*}
         D_{\textrm{\tiny KL}}(P_1 \mid \mid P_2) = \begin{cases}
             \int p(x) \log \left( \frac{p_1(x)}{p_2(x)}\right)d \mu(x), &\text{ if } Q \ll P,\\
             +\infty &\text{ otherwise}.
         \end{cases}
     \end{align*}

Finally, when talking about scoring imputations, we will take $\| \cdot\|_{2}$ to be the Euclidean distance on $\R^d$ and write expectations as $\E_{\substack{X \sim H\\ Y \sim P^*}}[ \| X-Y \|_{2}]$, to clarify over which distributions the expectation is taken.

\section{Background}\label{Sec_Background}

We are interested in creating a reliable ranking method for imputations when the underlying complete data is not available. To this end, we consider the I-Scores framework of \citet{ImputationScores}. 

\begin{definition} [Definition 4.1 in \cite{ImputationScores}]  \label{Iscoredef}
A real-valued function $S_{\textrm{\tiny NA}}(H, P)$
is a proper I-Score if, for any imputation distribution $H\in \mathcal{H}_{P}$,
\begin{equation}\label{scorecondition}
S_{\textrm{\tiny NA}}(H, P) \leq S_{\textrm{\tiny NA}}(P^*, P). 
\end{equation}
It is strictly proper if the inequality is strict for all $H\in \mathcal{H}_{P}$ such that $H \neq P^*$.
\end{definition}

We note that the second argument of the I-Score is the observed data distribution $P$ and not the true distribution $P^*$, thus highlighting that $S_{\textrm{\tiny NA}}(H, P)$ can be computed or at least approximated using observed data. Nonetheless, it should score $P^*$ highest, i.e. \eqref{scorecondition} must hold. This is only possible under assumptions on the missingness mechanism. 


We now present the missingness mechanisms relevant for this paper and then turn to the difficulty of scoring under MAR.

\subsection{MAR mechanisms}

We first discuss several MAR mechanisms, using the framework of pattern mixture models (PMM, \cite{little_patternmixture}). In PMM, one factors the joint distribution $p^*(x,m)$ as $p^*(x \mid M=m) \Prob(M=m)$. This factorization emphasizes that each observation is a masked draw from the distribution $X \mid M=m$ and allows for definitions of MAR that are more suitable for analyzing imputations. We refer to \citet{näf2024goodimputationmarmissingness} and the references therein for further details.

First, the usual MAR condition can be written in PMM notation  \citep[see, e.g.,][]{ourresult, näf2024goodimputationmarmissingness}. 

\begin{definition}[PMM-MAR]
    The missingness mechanism is missing at random (MAR) if, for all $m \in \mathcal{M}$ and $x \in \X_{\mid m} $, 
\begin{align}
&p^*(o^c(x ,m ) \mid o(x ,m ), M =m )= p^*(o^c(x ,m )\mid o(x ,m )). \tag{PMM-MAR}\label{PMMMAR}  
\end{align}
\end{definition}

This definition of MAR highlights the difficulty of imputation. In particular, \ref{PMMMAR} does not constrain $p^*(o^c(x ,m ) \mid o(x ,m ), M =m' )$ for any pattern $m' \neq m$. However, for imputation, one needs to learn $p^*(o^c(x ,m ) \mid o(x ,m )) $ from other patterns, such as $m=0_d$, where $0_d$ is the $d$-dimensional zero vector. As such, simpler conditions were also studied.

\begin{definition}
Assume $0_d \in \mathcal{M}$. The missingness mechanism is Extended Missing At Random (EMAR), if, for all $m \in \mathcal{M}$, for all $m' \in \{0_d, m\}$, and for all $x \in \X_{\mid m} \cap \X_{\mid 0}$, we have
    \begin{align} 
    &p^*(o^c(x ,m ) \mid o(x ,m ), M =m' )= p^*(o^c(x ,m )\mid o(x ,m ))\tag{EMAR} \label{EMAR}.
\end{align}
\end{definition}

\begin{definition}
    The missingness mechanism is conditionally independent MAR (CIMAR) if, for all $m, m' \in \mathcal{M}$ and $x \in \X_{\mid m'} $, we have
    \begin{align}
&p^*(o^c(x ,m ) \mid o(x ,m ), M =m' )= p^*(o^c(x ,m )\mid o(x ,m )) \tag{CIMAR} \label{CIMAR}.  
\end{align}
\end{definition}
We may equivalently write \ref{EMAR} (respectively \ref{CIMAR}) as
\begin{align*}
    p^*(o^c(x ,m ) \mid o(x ,m ), M =m )=p^*(o^c(x ,m ) \mid o(x ,m ), M =m' ),
\end{align*}
for $m'=0_d$ (respectively for all $m' \in \mathcal{M}$). Thus, these conditions ensure that the imputation distribution $p^*(o^c(x ,m )\mid o(x ,m ))$ can be learned from the fully observed pattern (\ref{EMAR}) or any other pattern (\ref{CIMAR}). If $0_d \in \mathcal{M}$, \ref{CIMAR} implies \ref{EMAR}. Finally, if \ref{PMMMAR} does not hold, the missingness mechanism is referred to as missing not at random (MNAR).

Note that all MAR definitions introduced so far allow for arbitrary distribution shifts in observed variables in general. For instance, two variables $(X_1^*, X_2^*)$ can have bounded and disjoint support in two different patterns $m_1=(0,0)$ and $m_2=(0,1)$. As long as the relationship and thus the conditional distribution $P^*_{X_1 \mid X_2}$ is the same in both patterns, \ref{CIMAR} and, in particular, \ref{PMMMAR} are satisfied. To avoid problems associated with extrapolation, we take the following assumption throughout.

\begin{assumption}\label{overlap}
    For all $m, m' \in \mathcal{M}$, $\X_{\mid m}=\X_{\mid m'}$.
\end{assumption}
This assumption still allows for various distribution shifts in the observed variables, as shown in and Example in Section \ref{Sec_Gaussmixmodel}.
We now introduce Example 4 in \citet{näf2024goodimputationmarmissingness} which serves as the leading example throughout the paper. The example was originally introduced to demonstrate that \ref{EMAR} and \ref{CIMAR} are strictly stronger assumptions than \ref{PMMMAR}.

\begin{example}\label{Example1}

Consider
\begin{align*}
\mathcal{M}= \{ m_1, m_2, m_3 \}=\left\{\begin{pmatrix} 0 & 0 & 0\end{pmatrix}, \begin{pmatrix}0 & 1 & 0\end{pmatrix}, \begin{pmatrix} 1 & 0 & 0 \end{pmatrix} \right \},
\end{align*}
and $(X^*_1, X^*_2, X^*_3)$ independently uniformly distributed on $[0,1]$ (later we will also consider dependence between the three variables). We further specify that 
\begin{align*}
    \Prob(M=m_1 \mid x) &= \Prob(M=m_1 \mid x_1) = x_1/3\\
    \Prob(M=m_2 \mid x) &= \Prob(M=m_2 \mid x_1) = 2/3-x_1/3\\
    \Prob(M=m_3 \mid x) &= \Prob(M=m_3) = 1/3.
\end{align*}

It can be checked that 
the MAR condition \ref{PMMMAR} holds. In particular, for variable $x_1$ in pattern $m_3$, it holds that
    \begin{align*}
        p^*(x_1 \mid x_2, x_3, M=m_3) =  p^*(x_1 \mid x_2, x_3).
    \end{align*}
    However, if we consider $x_1$ given $(x_2,x_3)$ in the first pattern, we have:
    \begin{align*}
        p^*(x_1 \mid x_2, x_3, M=m_1)&= \frac{\Prob(M=m_1 \mid x_1, x_2, x_3)}{\Prob(M=m_1 \mid x_2,x_3)} p^*(x_1 \mid x_2, x_3)\\
        &=2x_1 p^*(x_1 \mid x_2,x_3),
    \end{align*}
    showing that both \ref{CIMAR} and \ref{EMAR} do not hold. Indeed, \ref{PMMMAR} allows for a change in the conditional distributions over different patterns and requires only that the distribution $X^*_1 \mid X^*_2, X^*_3$ in pattern $m_3$ corresponds to the unconditional one. Thus the ideal imputation for both $X_1$ and $X_2$ ($P^*$) simply consists of drawing from the uniform distribution. In the following, we would like to find a proper I-Score that assigns the highest score to this imputation.
\end{example}

\subsection{The difficulty of scoring under MAR}\label{Sec_Scoringdifficulty}

The notion of I-Scores was motivated by proper scoring rules in prediction (see e.g, \citet{Gneiting2008} for a thorough treatment). In particular, we will focus here on the energy score. Let $P_1$ be a predictive distribution on $\R^d$ and $y$ be a test point drawn from a distribution $P_2$ on $\R^d$. Then, the energy score
    \begin{align*}
        es(P_1,y) = \frac{1}{2} \E_{\substack{X \sim P_1\\X' \sim P_1}}[| X-X' |]- \E_{X \sim P_1}[ | X - y |],
    \end{align*}
is a strictly proper score, in the sense that, $\E_{Y \sim P_2}[  es(P_1,Y)] < \E_{Y \sim P_1}[  es(P_1,Y)]$, for any $P_2 \neq P_1$. We use the energy score below to construct proper I-Scores.

To illustrate the difficulties of constructing proper I-scores to evaluate the quality of imputations,  let us consider the following simple example. Let $\mathcal{M}=\{m_1,m_2\}=\{(0,0), (1,0)\}$ ($X_1$ missing in the second pattern). With the assumption of equal support (\Cref{overlap}), \ref{PMMMAR} holds if for all $x \in \X$,
$$p^*(x_1 \mid x_2, M=m_1)=p^*(x_1 \mid x_2)=p^*(x_1 \mid x_2, M=m_2).$$
In this example, \ref{PMMMAR} is equivalent to \ref{CIMAR}, though this is not true in general, as Example \ref{Example1} shows. Nonetheless, the joint distribution of $(X_1^*,X_2^*)$ can still shift from pattern $m_1$ to pattern $m_2$. Assume we know the true $P^*$, then the equation above tells us that one can impute $X_1$ with the distribution $P^*_{X_1 \mid X_2}$. In order to evaluate the quality of this imputation, one often compares the imputed values in pattern $m_2$ to the observed values in pattern $m_1$. However, due to possible shifts of the marginal $X_2^*$ between these two patterns, the distribution of the imputed values and that of the observed values in pattern $m_1$ may differ. With such an evaluation process, the correct imputation $P^*_{X_1 \mid X_2}$ appears to be of poor quality. In other words, a score constructed with such an approach (comparing imputed values in a pattern with observed values in another pattern) would certainly be improper. 

Instead of comparing these two distributions, one could  compare the conditional distributions $P^*_{X_1 \mid X_2, M=m_1}$ and $P^*_{X_1 \mid X_2, M=m_2}$. This was the approach taken with the DR-I-Score in \citet{ImputationScores}. Alternatively, one may artificially mask some $X_1$ in pattern $m_1$, impute with $H \in \mathcal{P}$, and compare the resulting imputation with the original values of $X_1$.
This is the approach studied in this paper. In general, we may hope that this ``conditional'' approach leads to proper I-Scores under \ref{PMMMAR} because of the following result.

\begin{restatable}{proposition}{identificationprop}[Adaptation of Proposition 2.4 in \citet{näf2024goodimputationmarmissingness}]\label{amazingprop}
    Under \ref{PMMMAR}, for all $j \in \{1, \ldots, d\}$,
    \begin{align}\label{conditionalindepunderMAR}
        P^*_{X_j \mid X_{-j}, M_j=0}=P^*_{X_j \mid X_{-j}}
    \end{align}
\end{restatable}
To illustrate this, we consider Example \ref{Example1}. Here, the conditional distributions of $X_1 \mid X_2, X_3$ in each of the two patterns in which $X_1$ is observed are different from the pattern where $X_1$ is missing. Nonetheless it can be shown that,
\begin{align*}
    &p^*(x_1 \mid x_2,x_3, M_1=0)\\
     &= \frac{\Ind\{(x_2,x_3) \in [0,1]^2\} }{\Ind\{(x_2,x_3) \in [0,1]^2\}}\frac{1}{4} 2x_1 \Ind\{x_1 \in \{0,1\}\} +\frac{\Ind\{(x_2,x_3) \in [0,1]^2\}}{\Ind\{(x_2,x_3) \in [0,1]^2\}} \frac{1}{2}\cdot(2-x_1) \Ind\{x_1 \in \{0,1\}\}\\
     &=\Ind\{x_1 \in [0,1]\},
\end{align*}
which corresponds to $p^*(x_1 \mid x_2,x_3)$, the correct imputation distribution under \ref{PMMMAR}. Now assume we have access to $P^*_{X_{-1} \mid M_1=0}$, allowing to impute $X_{-1}^*=(X_2^*,X_3^*)$ for all patterns such that $M_1=0$, in a first step. In this case, masking  $X_1^*$ when it is observed ($M_1=0$), imputing with $H_{X_1 \mid X_{-1}}$ and evaluating the imputation with a proper scoring method, would to the correct imputation $P^*_{X_1 \mid X_{-1}}$ having the highest score by \eqref{conditionalindepunderMAR}. 


Naturally, having access to $P^*_{X_{-1} \mid M_1=0}$ is unrealistic and in fact violates Definition \ref{Iscoredef}, as the score can only be based on $P$. As such, the task of constructing a proper I-score becomes more difficult. 
Consider an approach wherein we directly impute all missing values for $m \in L_j$ simultaneously with $X_j$, i.e., by imputing $(X_j,o^c(X,m))$, and then score only $X_j$. This idea is studied in more detail with an alternative score in Appendix \ref{Sec_noprojectionsscore}. In Example \ref{Example1}, this would mean that in pattern $m_1\in L_1$, only $X_1$ would be imputed, while in pattern $m_2 \in L_1$, $(X_1,X_2)$ would be imputed simultaneously. However, in this example, artificial masking of $X_1$ in patterns $m_1$ and $m_2$ will lead to an MNAR mechanism, where \ref{PMMMAR} no longer holds. Consider for instance pattern $m_2$: As \ref{PMMMAR} holds, $p^*(x_2 \mid x_1, x_3, M=m_2)=p^*(x_2 \mid x_1, x_3)$, and thus drawing from $P^*_{X_2 \mid X_1, X_3}$ in pattern $m_2$ is valid. Once we artificially set $X_1$ missing, we however have to impute $X_1$ along $X_2$ and it can be checked that in this case $p^*(x_1, x_2 \mid x_3, M=m_2) \neq p^*(x_1,x_2 \mid x_3) $. Similarly, in pattern $m_1$, $p^*(x_1 \mid  x_2, x_3, M=m_1)=2x_1\Ind\{0\leq x_1 \leq 1\} \neq p^*(x_1 \mid x_2, x_3)$. Thus, drawing from $P^*_{X_1 \mid X_2, X_3}$ in pattern $m_1$ and from $P^*_{X_1, X_2 \mid X_2, X_3}$ in pattern $m_2$, $P^*$ will seem like a bad imputation if the draws are compared directly to the observed values. The imputation method that scores highest, thus becomes $H \in \mathcal{P}$, where $H$ imputes $P^*_{X_1 \mid X_2, X_3, M=m_1}$ in pattern $m_1$ and $P^*_{X_1,X_2 \mid X_3, M=m_2}$ in pattern $m_2$. This is true independently of how $H$ actually imputes, that is, even if $H_{X_1 \mid X_2,X_3, M=m_3}\neq P^*_{X_1 \mid X_2,X_3}$; its score will nonetheless be strictly higher than the one of $P^*$. At the same time, we also note that this is not the same as only considering the fully observed pattern (pattern $m_1$), as is done by the DR-I-Score (see Section \ref{Sec_DRIscore}). In this case, $2x_1p^*(x_1)$ would be wrongly judged to be the best imputation.




While further research might uncover better ways to construct a score than testing the imputation of a variable $X_j$ given the observed variables, the fundamental problem remains: masking observed values under \ref{PMMMAR} will in general lead to  MNAR mechanisms, which in turn will lead to poor scoring of the optimal imputation method (based on the original  \ref{PMMMAR} mechanism). One could resort to more stringent assumptions, that guarantee \ref{PMMMAR} even when $X_j$ is masked (see Appendix \ref{Sec_noprojectionsscore}). In the following, we instead study a different approach that will lead to a proper score under Example \ref{Example1}.

\section{Energy-I-Score: a proper scoring method} \label{Sec_Scoring}


Given the difficulties of scoring under MAR described in the last section, we study alternative assumptions. Let $O$ be the set of indices of fully observed variables,
\begin{align}\label{Odef}
    O=\{j: m_j = 0 \text{ for all } m \in \mathcal{M} \},
\end{align}
and, assume that $O \neq \emptyset$ and that for all $x \in \mathcal{X}$, $m \in \mathcal{M}$:
\begin{align}\label{RMAR}
    \Prob(M=m \mid x)=\Prob(M=m \mid x_{O})\tag{RMAR}.
\end{align}
This widely used condition \citep{RMAR_Psychometrika, RMAR2006, Mohan2013} may appear more explicit than \ref{PMMMAR} as it involves components that are always observed. However, it is  equivalent to \ref{CIMAR}, as shown in \cite{näf2024goodimputationmarmissingness}, and the question remains whether a weaker condition is possible. As such, we introduce a new condition that is motivated by both \ref{RMAR} and \eqref{conditionalindepunderMAR}. We recall the definition of $L_j$ as the set of patterns $m$ such that $m_j=0$. Consequently $L_j^c$ is the set of patterns $m$ in which variable $j$ is missing.

\begin{definition}
 For all $j \in \{1, \hdots, d\}$ such that $L_j^c \neq \emptyset$, $L_j \neq \emptyset$, let $
    O_j = \bigcap_{m \in L_j} \{l: m_l=0\}$ be the set of indices of variables that are always observed when the $j$th variable is observed. 
\end{definition}

\begin{definition}
    We say that Condition CIMAR$_j$ holds if  
    \begin{align}
L_j^c \neq \emptyset, L_j \neq \emptyset, O_j \neq \emptyset \text{ and } P^*_{X_j \mid X_{O_j}, M_j=0}=P^*_{X_j \mid X_{O_j}} \tag{CIMAR$_j$}\label{ass_score_j}  
\end{align}
\end{definition}

\begin{assumption}
\label{ass_score}
There exists $j \in \{1, \ldots, d\}$ such that Condition \ref{ass_score_j} holds.
\end{assumption}


Assumption \ref{ass_score} is a direct adaptation of \eqref{conditionalindepunderMAR}. If a variable $j$ is only observed when all other variables are observed, i.e. in pattern $m=0_d$, then $X_{O_j}=X_{-j}$ and we recover \eqref{conditionalindepunderMAR}. Moreover, if \ref{RMAR} holds, $X_{O_j}=X_O$ for all $j$ with missing observations.

Consider again Example \ref{Example1}. Here for $j=1$ and $j=2$, it holds that $O_j=\{3\}$. Moreover, we have that
\begin{align*}
    \Prob(M_1=0 \mid x_1,x_3)=\Prob(M=m_1 \mid x_1) +\Prob(M=m_2 \mid x_1) = \Prob(M_1=0).
\end{align*}
Thus, $X_1 \indep M_1 \mid X_3$ and Assumption \ref{ass_score_j} holds for $j=1$. Similarly, 
\begin{align*}
    \Prob(M_2=0 \mid x_2,x_3)=\int \Prob(M=m_2 \mid x_1) p^*(x_1 \mid x_2,x_3) d x_1 = \Prob(M=m_2),
\end{align*}
again showing $X_2 \indep M_2 \mid X_3$ and that Assumption \ref{ass_score_j} holds for $j=2$. Thus, while Assumption \ref{ass_score} holds whenever \ref{RMAR} holds as discussed above, Example \ref{Example1} is a case for which Assumption \ref{ass_score} holds but not \ref{RMAR}.

We note that for $j=2$ this is only true by independence of $X_1$ from $(X_2,X_3)$. In particular, if we introduce dependence between $X_1$ and $X_2$ in this example (e.g., through a copula), \ref{PMMMAR} still holds, yet $\Prob(M_2=0 \mid x_2,x_3)= \Prob(M=m_2 \mid x_2)$. Thus this modified version of Example \ref{Example1} shows that there are cases that are \ref{PMMMAR} but do not meet Assumption \ref{ass_score}. Maybe surprisingly, another modification of Example \ref{Example1} in Appendix \ref{Sec_Proofs} shows that there are cases that meet Assumption \ref{ass_score} but that are not \ref{PMMMAR}. Thus, while Assumption \ref{ass_score} is a direct adaptation of \eqref{conditionalindepunderMAR}, itself a consequence of \ref{PMMMAR}, neither condition implies the other.

\subsection{Population score}\label{Sec_Scoring_1}


The (expected) energy score that we will use scores a distributional prediction by comparing it to a test point. Taking the expectation over a set of test points, one can show that the score is maximized if the predictive distribution agrees with the actual distribution of the test points as described in Section \ref{Sec_Background}. The idea is now to score the imputation distribution $H_{X_j \mid x_{O_j}, M_j=1}$ by comparing it to a point drawn from $P^*_{X_j \mid x_{O_j}, M_j=0}$, which will be simply an observed point. That is, for all $j \in \{1, \ldots, d\}$, we define the score of the variable $j$ as
\begin{align}\label{Scoreforj}
  &S^j_{\textrm{\tiny NA}}( H, P ) = \nonumber\\
  &\E_{X_{O_j} \sim P_{X_{O_j}\mid M_j=0}}\Big[   \E_{\substack{X \sim H_{X_j \mid X_{O_j}, M_j=1}\\Y \sim P^*_{X_j \mid X_{O_j}, M_j=0}}}[ \| X - Y \|_{2}] - \frac{1}{2} \E_{\substack{ X \sim H_{X_j \mid X_{O_j}, M_j=1}\\ X' \sim H_{X_j \mid X_{O_j}, M_j=1}}}[\| X-X' \|_{2}] \Big] ,
\end{align}
where the outer expectation is taken over $X_{O_j} \sim P_{X_{O_j}\mid M_j=0}$, the distribution of all fully observed variables wrt to variable $j$.\footnote{Usually, in the scoring literature, one only considers the inner expectation, even though in practice ``scores are reported as averages over comparable sets of probabilistic forecasts'' \citep[page 222]{Gneiting2008}. We thus also consider the outer expectation to model the different test points.} Note that by assumption $X_{O_j}$ is observed for $M_j=0$, such that, $P_{X_{O_j}\mid M_j=0}=P^*_{X_{O_j}\mid M_j=0}$. Then, the full score is given as
\begin{align*}
    S_{\textrm{\tiny NA}}(H, P) = \frac{1}{|\mathcal{S}|} \sum_{ j \in \mathcal{S}} S_{\textrm{\tiny NA}}^j(H, P) ,
\end{align*}
where $\mathcal{S}=\{j: L_j^c \neq \emptyset, O_j \neq \emptyset\}$
is the set of variables $j$ that can be missing and such that at least one other variable is always observed when the $j$th variable is observed. 

\begin{restatable}{proposition}{proprietyprop}\label{proprietyprop}
    Under \Cref{ass_score} and \ref{PMMMAR},  $ S_{\textrm{\tiny NA}}$ is a proper I-Score. 
\end{restatable}



Even if \ref{ass_score_j} holds for all $j \in \{1,\ldots, d\}$, the score would not be \textit{strictly} proper. For example, consider three variables $X_1, X_2, X_3$ and assume that $X_3$ is fully observed, while either $X_1$ or $X_2$ are missing (i.e., $\mathcal{M}=\{(0,1, 0), (1,0, 0)\}$). In this case, even if \ref{ass_score_j} holds for $j=1$, the imputation $\tilde{H}$ which has $\tilde{H}_{X_1 \mid X_3}=P^*_{X_1 \mid X_3}$ and $\tilde{H}_{X_2 \mid X_3}=P^*_{X_2 \mid X_3}$, but imputes such that $X_1 \indep X_2 \mid X_3$, i.e. $\tilde{H}_{(X_1, X_2) \mid X_3} = \tilde{H}_{X_1 \mid X_3} \times \tilde{H}_{X_2 \mid X_3} $, will achieve the same score as $P^*$, even if $P^*_{(X_1, X_2) \mid X_3} = P^*_{X_1 \mid X_3} \times P^*_{X_2 \mid X_3}$ does not hold.

\subsection{Score estimation}\label{Sec_Scorest}


We now describe an estimation strategy for $S_{\textrm{\tiny NA}}(H, P)$ based on a sample of $n$ observations with missing values. To differentiate in the following between the index of the observation $i$ and the index of the dimension $j$, we will write $X_{i, \cdot}/M_{i, \cdot}$ for observation $i$, such that we observe a sample $(X_{1, \cdot}, M_{1, \cdot}), \ldots, (X_{n, \cdot}, M_{n, \cdot})$. 

Fix $j \in \mathcal{S}$ and let $i$ such that $m_{i, \cdot} \in L_j$, which implies, by definition that $x_{i,j}$ is observed. 
We build a sample of $N$ points, $\tilde{X}_{i,j}^{(1)},  \ldots, \tilde{X}_{i,j}^{(N)}$ by imputing $x_{i,j}$ $N$ times, via sampling from $H_{X_j \mid X_{O_j}}$, as follows: 
\begin{enumerate}
    \item Create a new data set by concatenating the observed $(x_{i,j}, x_{i,O_j})$, $m_{i, \cdot} \in L_j$ and the imputed $(x_{i,j}, x_{i,O_j})$, $m_{i, \cdot} \in L_j^c$, as in Figure \ref{fig:scoreillustrationOj}, and set the \emph{observed} observations of $X_{j}$ to missing, i.e. $x_{i,j}=\texttt{NA}$ for $i$ with $m_{i, \cdot} \in L_j$.
\item Approximate the sampling from $H_{X_j \mid X_{O_j}, M_j=1}$ by simply imputing these artificially created NA values with $H$, $N$ times.
\end{enumerate}
As multiple imputation corresponds to drawing several times from the corresponding conditional distribution, this is a natural way of obtaining $\tilde{X}_{i,j}^{(l)}$, $l=1,\ldots, N$. If a method is unable to generate multiple imputations, $\tilde{X}_{i,j}^{(l)}$ is just a unique value copied $N$ times.

We can use the generated samples $\tilde{X}_{i,j}^{(1)},  \ldots, \tilde{X}_{N}^{(N)}$, 
in order to  estimate $S^j_{NA}( H, P )$:
\begin{align}\label{newscore}
    \widehat{S}^j_{\textrm{\tiny NA}}( H, P )= \frac{1}{|\{i: m_{i, \cdot} \in L_j\}|} \sum_{{i: m_{i, \cdot} \in L_j}}  \left(\frac{1}{2N^2} \sum_{l=1}^N \sum_{\ell=1}^N | \tilde{X}_{i,j}^{(l)} - \tilde{X}_{i,j}^{(\ell)}    | -  \frac{1}{N} \sum_{l=1}^N | \tilde{X}_{i,j}^{(l)} - x_{i,j} |   \right),
\end{align}
as in \citet[Equation (7)]{Gneiting2008}. This is nothing more than the empirical counterpart of Equation \eqref{Scoreforj}. 
The final score is then given as 
\begin{align}\label{finalformulascore}
   \widehat{S}_{\textrm{\tiny NA}}(H, P) =\frac{1}{|\mathcal{S}|} \sum_{ j \in \mathcal{S}} \widehat{S}_{{\textrm{\tiny NA}}}^j(H, P).
\end{align}

We refer to this approach as the energy-I-Score, referencing the use of the energy distance. The energy-I-Score score thus uses the ability of imputation methods to generate multiple imputations naturally in its scoring. Unfortunately, this can be computationally demanding, depending on the size of $N$. In our experiments in Sections \ref{Sec_Simulation}, \ref{Sec_Empirical} we chose $N=50$. Moreover, if all variables contain missing values, this would need to be repeated $d$ times. This would be infeasible for realistic dimensions if the full data set had to be imputed each time. However, note that in the data set created in Steps 1. and 2. only one variable has missing values, while all the others are observed. For instance, for an FCS method, this means that only one iteration. Moreover, for large $d$ it is possible to only consider a subset of variables $X_j$ to calculate $\hat{S}_{NA}$. For example, one could choose the $p < d $ variables with the highest missingness proportion. Finally, Appendix \ref{App_ChoiceofN} studies the effect of $N$ on our examples in more detail and shows that the score is also still accurate for smaller $N$, such as $N=20$. Appendix \ref{Sec_Details} provides further implementation details, including the handling of mixed data, and provides pseudo-code in Algorithm \ref{algorithm1}.

\begin{figure}
    \centering

    \begin{tikzpicture}[scale=0.8, every node/.style={scale=0.8},
    box/.style={draw, rectangle, minimum width=2.5cm, minimum height=1.8cm, align=center},
    smallbox/.style={draw, rectangle, minimum width=2cm, minimum height=1cm, align=center},
    bluebox/.style={draw=blue, rectangle, minimum width=2cm, minimum height=1cm, align=center},
    arrow/.style={-{Stealth[length=5pt]}, blue, thick}
]

\node[box] (observed) at (-0.5,3) {Observed\\Values};
\node[box, draw=blue] (imputed) at (-0.5,0) {Imputed\\Values};

\matrix (data) [matrix of math nodes, row sep=0.2cm, column sep=0.8cm, nodes={minimum width=1.5cm}, font=\large] at (4,1.5) {
  x_{1,j} & x_{1,O_j} \\
  x_{2,j} & x_{2,O_j} \\
  \vdots & \vdots \\
  x_{\ell,j} & x_{\ell,O_j} \\
  x_{\ell+1,j} & x_{\ell+1,O_j} \\
  \vdots & \vdots \\
};

\node at ($(data.north west)+(-0.5,-2.7)$) {\scalebox{2.2}{$\left(\vphantom{\begin{matrix}x_{1,j}\\x_{2,j}\\\vdots\\x_{\ell+1,j}\\\vdots\end{matrix}}\right.$}};
\node at ($(data.north east)+(0.5,-2.7)$) {\scalebox{2.2}{$\left.\vphantom{\begin{matrix}x_{1,j}\\x_{2,j}\\\vdots\\x_{\ell+1,j}\\\vdots\end{matrix}}\right)$}};

\node[box] (score) at (4,5.5) {Score\\with\\$x_{i,j}$ as a\\test point};
\node[box] (generate) at (9,3) {Generate\\$(X_1^{i},...,X_N^{i})$\\from $H_{X_j \mid x_{i,O_j}, M_j=1}$};
\node[bluebox] (learn) at (9,0.5) {Learn\\$H_{X_j \mid x_{i,O_j}, M_j=1}$};

\draw[arrow] ($(data.east)+(0,-1)$) to[out=0,in=180] (learn);
\draw[arrow] (learn) -- (generate);
\draw[arrow] (generate) to[out=90,in=0] (score);

\draw[blue, thick] ($(data.west)-(-0.3,0)$)-- ++(-0.2,0) -- ++(0,-2.6) -- ++(0.2,0);
\draw[blue, thick] ($(data.east)+(-0.3,0)$)-- ++(0.2,0) -- ++(0,-2.6) -- ++(-0.2,0);

\draw[black, thick] ($(data.west)-(-0.3,-2.6)$)-- ++(-0.2,0) -- ++(0,-2.6) -- ++(0.2,0);
\draw[black, thick] ($(data.east)+(-0.3,2.6)$)-- ++(0.2,0) -- ++(0,-2.6) -- ++(-0.2,0);

\end{tikzpicture}

    \caption{Conceptual illustration of the score approximation. First, the imputed values in blue are used to learn $H_{X_j \mid X_{O_j}, M_j=1}$. Then, for each $x_{i,O_j}$ for which $x_{i,j}$ is observed, we score the ``prediction'' $H_{X_j \mid X_{O_j}, M_j=1}$ using the energy score with test point $x_{i,j}$. In practice, this is done by (approximately) generating a sample $\tilde{X}_{i,j}^{(l)}, l=1,\ldots, N$ from $H_{X_j \mid x_{i,O_j}, M_j=1}$.}
    \label{fig:scoreillustrationOj}
\end{figure}
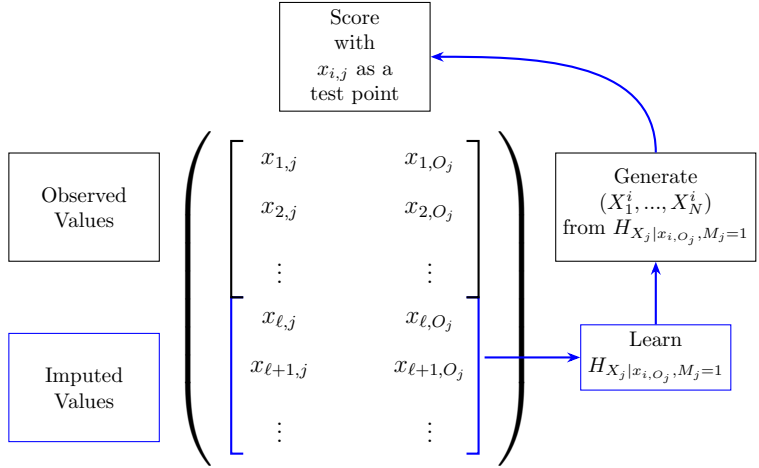

\section{Revisiting the DR-I-Score}

 \label{Sec_DRIscore}


\citet{ImputationScores} developed a first I-Score using the KL divergence and random projections $A \subset \{1,\ldots, d\}$, denoted ``DR-I-Score''. The set of projections could be chosen by a practitioner as a tool to increase the number of fully observed observations. The estimation procedure of the score is somewhat complicated, so here we focus on the essentials and refer to \citet{ImputationScores} for details.  

The DR-I-Score ranks any imputation $H$  by comparing the imputed distribution in a pattern $m$, $H_{X \mid M=m}$, with the distribution on the  fully observed pattern $P^*_{X \mid M=0}$ using the KL divergence 
$$D_{\textrm{\tiny KL}}(H_{X \mid M=m} \mid \mid P^*_{X \mid M=0}).$$ 
Specifically, given a distribution of possible projections $\mathcal{K}$, the population version of the DR-I-Score is defined as: 
       \begin{equation*}
           S_{\textrm{\tiny NA}}^{\textrm{\tiny DR}}(H,P) = - \E_{A \sim  \mathcal{K}, \tilde{M}_A \sim P_A^M}[D_{\textrm{\tiny KL}}(H_{X_A \mid M_A=\tilde{M}_A} \mid \mid P^*_{X_A \mid M_A=0})],
       \end{equation*}
where for a projection $A$, $M_A=(M_j)_{j \in A}$, $X_A=(X_j)_{j \in A}$, $P_A^M$ is the distribution of $M_A$, $H_{X_A \mid M_A=\tilde{M}_A}$ is the imputation distribution projected to $A$ conditional on the pattern $\tilde{M}_A$ and $P^*_{X_A \mid M_A=0_{|A|}}$ is the correct distribution of $X_A$ conditional on the pattern $M_A=0_{|A|}$. Thus for a projection $A \sim \mathcal{K}$, the score compares an imputed pattern $\tilde{M}_A$ with the fully observed pattern on that projection using the KL divergence. Due to the density ratios involved, the KL divergence is difficult to estimate in practice from multivariate samples. To address this, \citet{ImputationScores} use a classifier to approximate the density ratios and estimate the KL divergence, similar to \citet{Cal2020} in the context of two-sample testing. In doing so, a training and test set is defined on each projection $A$ and for each pattern $M_A$. Moreover, to make the score viable in practice, the set of projections $\mathcal{K}$ was, in turn, determined by the available patterns. This added an additional layer of complexity in the implementation in \cite{Iscorespackage}, as outlined in \citet[Section 5]{ImputationScores}.

\subsection{Under which conditions is the DR I-Score proper?}

\citet{ImputationScores} proved that the DR-I-Score is a proper I-Score if \ref{CIMAR} holds on \emph{each projection $A$}. However, inspecting the proof of their Proposition 4.1 reveals that the DR-I-Score is actually proper under \ref{EMAR} on each projection. Considering $A=\{1,\ldots,d\}$ for simplicity, the proof relies on the fact that
\begin{align*}
    &D_{\textrm{\tiny KL}}(H_{X \mid M=m} \mid \mid P^*_{X \mid M=0})\\
                 &=\E_{o(X, m) \sim P^*_{X \mid M=m}} \Big[ \int \log \left( \frac{h(o^c(x, m)| o(x, m), M=m)}{p^*(o^c(x,m)| o(x,m))} \right) h(o^c(x, m)| o(x, m), M=m) \\
         & \quad \times d\mu(o^c(x,m)) \Big]
         +\int \log \left( \frac{ p^*(o(x,m)| M=m)}{ p^*(o(x,m) |  M=0)} \right)  p^*(o(x,m)| M=m) d\mu(o(x,m)).\\ 
\end{align*}
Thus, the calculation of the KL divergence between the two patterns $M=m$ and $M=0$ can be split into two parts, whereby the second part is an irreducible error not affected by any imputation. However, the first part is minimized when $H=P^*$, since under \ref{EMAR}, $p^*(o^c(x, m)| o(x, m), M=m)=p^*(o^c(x, m)| o(x, m), M=0)=p^*(o^c(x, m)| o(x, m))$. The same arguments made for each projection $A 
\sim \mathcal{K}$ proves that DR-I-Score is proper under \ref{EMAR} on each projection $A \sim \mathcal{K}$. Thus, even though distribution shifts of $o(X, m)$ are possible even under \ref{EMAR}, minimizing the above KL divergence amounts to retrieve the true correct conditional distribution $P^*_{X|M=0}$. 


\ref{EMAR} is a stronger condition than \ref{PMMMAR}, as shown in Example \ref{Example1}. In particular, the DR-I-Score is not proper in the setting of Example \ref{Example1}. To see this, let us evaluate how the correct imputation $H=P^*$ is ranked with the DR-I-Score. Consider the pattern $m_3=(1,0,0)$. The DR-I-Score compares the imputation distribution on pattern $m_3$, that is $p^*(x_1 \mid x_2, x_3) p^*(x_2, x_3)$  to the complete distribution in pattern $m_1$,
\begin{align*}
    p^*(x_1 , x_2, x_3 \mid M=m_1)=p^*(x_1 \mid x_2, x_3, M=m_1) p^*(x_2, x_3)=2x_1 p^*(x_1 , x_2, x_3),
\end{align*}
as $p^*(x_2, x_3 \mid M=m)=p^*(x_2, x_3)$ for all $m$.
Thus, while we would like to score the imputation $p^*(x_1, x_2, x_3)$ highest, imputing by $h(x_1, x_2, x_3) =2x_1 p^*(x_1 , x_2, x_3)$ will lead to a score value of exactly zero, while
\[
 D_{\textrm{\tiny KL}}(P^*_{X \mid M=m} \mid \mid P^*_{X \mid M = 0}) = \int p^*(x_1, x_2 , x_3) \log \left( \frac{1}{x_1} \right)  d\mu(x_1, x_2,x_3) > 0.
\]
We summarize this in the following Corollary:

\begin{corollary}\label{DR_Propriety}
    The DR-I-Score in \cite{ImputationScores} is proper if \ref{EMAR} holds on all projections $A$ in the support of $\mathcal{K}$. There exist examples that are \ref{PMMMAR} but not \ref{EMAR} such that the DR-I-Score is not proper. 
\end{corollary}


This discussion reveals two weaknesses of the DR-I-Score score. First, \ref{EMAR} on each projection might not be realistic. Second, and arguably more importantly, the KL-Divergence is not straightforward to estimate for multivariate observations. In fact, the projections are necessary as the number of fully observed points will be small for realistic datasets. In \citet{ImputationScores}, a classifier was used to provide an estimate of the KL divergence. If the classifier is not accurate, this can lead to undesirable rankings, as we demonstrate in Section \ref{Sec_Simulation}. This is in addition to the relative complex construction of the DR-I-Score. 


In contrast, instead of comparing pattern by pattern, the energy-I-Score works variable by variable, similar to the FCS approach, allowing for propriety also in settings such as in Example \ref{Example1}. This approach also naturally defines projections as $A=O_j$, instead of having to choose them somewhat arbitrarily. Moreover, it does not rely on a classifier and instead utilizes the energy score and sampling ability of the imputation methods. As such, instead of scoring (single) imputed datasets, the energy-I-Score naturally tests the ability of imputation methods to generate accurate samples. These properties are studied on simulated examples in the next section.



\section{Simulation Study} \label{Sec_Simulation}

We first demonstrate the abilities of our new score in simulated data, drawing inspiration from \cite{näf2024goodimputationmarmissingness} to design the following examples. First, we study a version of the leading example throughout the paper (Example \ref{Example1}), showcasing the performance of the score in a setting with a simple full data distribution, but complex (conditional) distribution shifts. Second, we study a Gaussian mixture example. Although this example is \ref{CIMAR}, it has strong distributional shifts in the observed variables, making it difficult to impute and score. The last example combines this with nonlinear relationships between the variables, showcasing an example where none of the imputation methods are ideal. Reproducible codes are available on GitHub at \url{github.com/KrystynaGrzesiak/ImputationScore}.



\paragraph{Imputation methods}  In each example, we rank the performance of the following FCS methods.
\begin{itemize}
    \item mice-cart - classification and regression trees \citep[see \textsf{R}-package \texttt{mice} ][]{mice}
    \item missForest - random forest  \citep[see \textsf{R}-package \texttt{missForest}][]{Stekhoven})
    \item mice-norm.predict - regression imputation   \citep[see \textsf{R}-package \texttt{mice}][]{mice}
    \item mice-norm.nob - Gaussian imputation \citep[see \textsf{R}-package \texttt{mice} ][]{mice}.
\end{itemize}

In all examples, we standardize the scores over the 10 repetitions so that they belong to $(-1,0)$.

\paragraph{Evaluation} 
First, we calculate the (negative) energy distance between the complete and imputed data sets, using the \texttt{energy} \textsf{R}-package \citep{energypackage}. As this ``score'' is based on the complete distribution, we refer to it as the full information score. We compare the orderings of the full information score with the energy-I-Score, which does not have access to the values underlying the missing values. The only hyperparameter to choose in this case is the number of samples $N$, which we set to $N=50$. Further choices of $N$ for the first two examples are explored in Appendix \ref{App_ChoiceofN}. We also compare the energy-I-Score with the DR-I-Score of \citet{ImputationScores}.


\paragraph{Results} The three examples considered in this section indicate that the ordering produced by the energy-I-Score is similar to the one of the full information score, even in the challenging distributional shift example in Section \ref{Sec_Gaussmixmodel}. In particular, it means that the energy-I-Score reliably finds the best imputation method, that is the one that has the highest full information score. An exception is the third example in Section \ref{Sec_Gaussmixmodelnonlinar} where none of the methods performs well. Here, the scores disagree quite heavily. This is to be expected as even in the population setting, propriety of the score only guarantees that the ideal method is maximizing the score in expectation. 
As with other metrics, there is no strictly correct ordering of non-optimal methods and different metrics will result in different rankings. For example, the energy-I-Score by construction tends to penalize methods that cannot produce multiple imputations. Thus while a method like missForest might produce an imputation that has a relatively small energy distance when compared to the complete data, it will produce $N$ identical imputation in the energy-I-Score procedure, resulting in a low score. Given the discussion in this and earlier papers \citep{VANBUUREN2018, ImputationScores,näf2024goodimputationmarmissingness}, this might be desirable.



\subsection{Uniform Example}\label{Sec_PMMMARExample}


We consider a setting similar to that of  \Cref{Example1}, where we add three variables. More precisely, 
we assume $X^*=(X^*_1, \ldots, X^*_6)$ are independently uniformly distributed on $[0,1]$ and specify the three patterns
\begin{align*}
    m_1 = (0, 0, 0, 0, 0, 0),\ \ m_2 = (0, 1, 0, 0, 0, 0),\ \ m_3 = (1, 0, 0, 0, 0, 0),
\end{align*}
with
\begin{align*}
    \Prob(M=m_1 \mid x) &= \Prob(M=m_1 \mid x_1) = x_1/3\\
    \Prob(M=m_2 \mid x) &= \Prob(M=m_2 \mid x_1) = 2/3-x_1/3\\
    \Prob(M=m_3 \mid x) &= \Prob(M=m_3) = 1/3.
\end{align*}
Since \ref{PMMMAR} but not \ref{EMAR} is true, trying to learn $p(x_1 \mid x_2, \ldots x_6)$ in the fully observed pattern leads to biased imputation, as shown in Figure \ref{fig:Application_1_Scores}. Moreover, as a consequence of Corollary \ref{DR_Propriety}, we expect the DR-I-Score to not be proper in this example. To illustrate this, we also add the imputation that draws from $p^*(x_1 \mid x_2, x_3, M=m_1)=2x_1 \Ind\{x_1 \in [0,1]\}$, which corresponds to the imputation discussed in Section \ref{Sec_DRIscore}, leading to a KL-Divergence of exactly zero. We denote this method by ``runifsq''. On the other hand, this example meets \Cref{ass_score}, as shown in Section \ref{Sec_Scoring_1} and thus we expect the energy-I-Score to be proper. To examine this, we also add an imputation drawing from the true distribution of $X_1$, i.e. uniform on $(0,1)$, denoted ``runif'' (corresponding to $P^*$). 

We draw $n=2'000$ i.i.d. copies of $(X^*, M)$, apply the missingness mechanism, and impute. Figure \ref{fig:Application_1_Scores} shows the results. As expected, the energy-I-Score recognizes that runif is the best method and produces an ordering similar to the full information score, without having access to the true underlying data. In contrast, the DR-I-Score scores the runifsq imputation highest, which empirically confirms that the score is not proper in this example.


In Appendix \ref{Sec_ExperimentsAppendix} we also consider the same example with dependence between $X_1, X_2, X_3$ and illustrate that even though Assumption \ref{ass_score} no longer holds, the energy-I-Score still identifies the best imputation.

\begin{figure}
    \centering
    \includegraphics[width=0.9\linewidth]{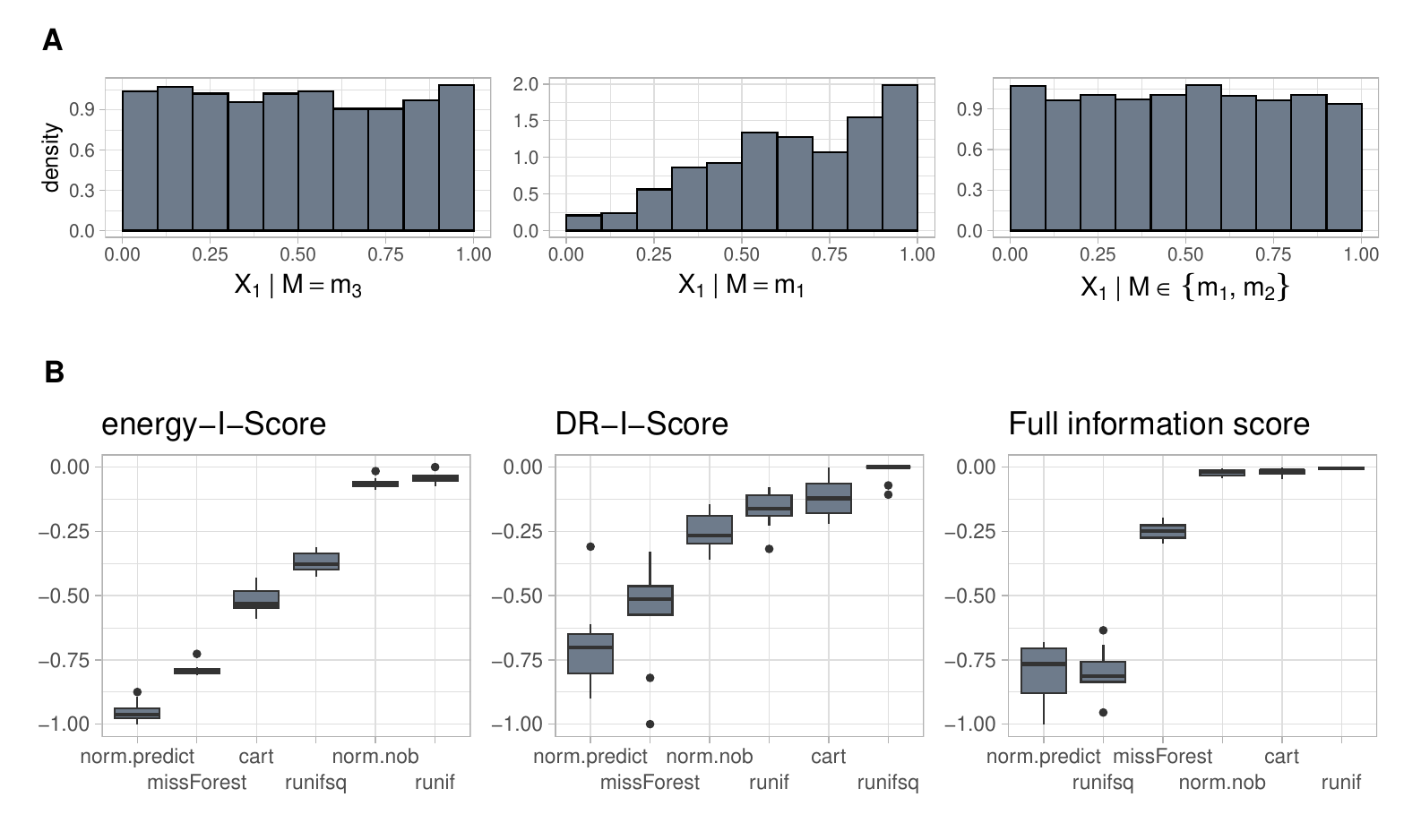}
    \caption{ (A) Illustration of Example \ref{Example1}. Left: Distribution we would like to impute $X^*_1 \mid M=m_3$. Middle: Distribution of $X_1$ in the fully observed pattern $(X_1 \mid M=m_1)$. Right: Distribution of all patterns for which $X_1$ is observed (Mixture of the distribution of $X_1$ in patterns $m_1$ and $m_2$). (B) Standardized scores for different imputations methods. Methods are ordered according to the mean score. }
    \label{fig:Application_1_Scores}
\end{figure}

\subsection{Gaussian mixture model} \label{Sec_Gaussmixmodel}

\begin{figure}
    \centering
    \includegraphics[width=0.9\linewidth]{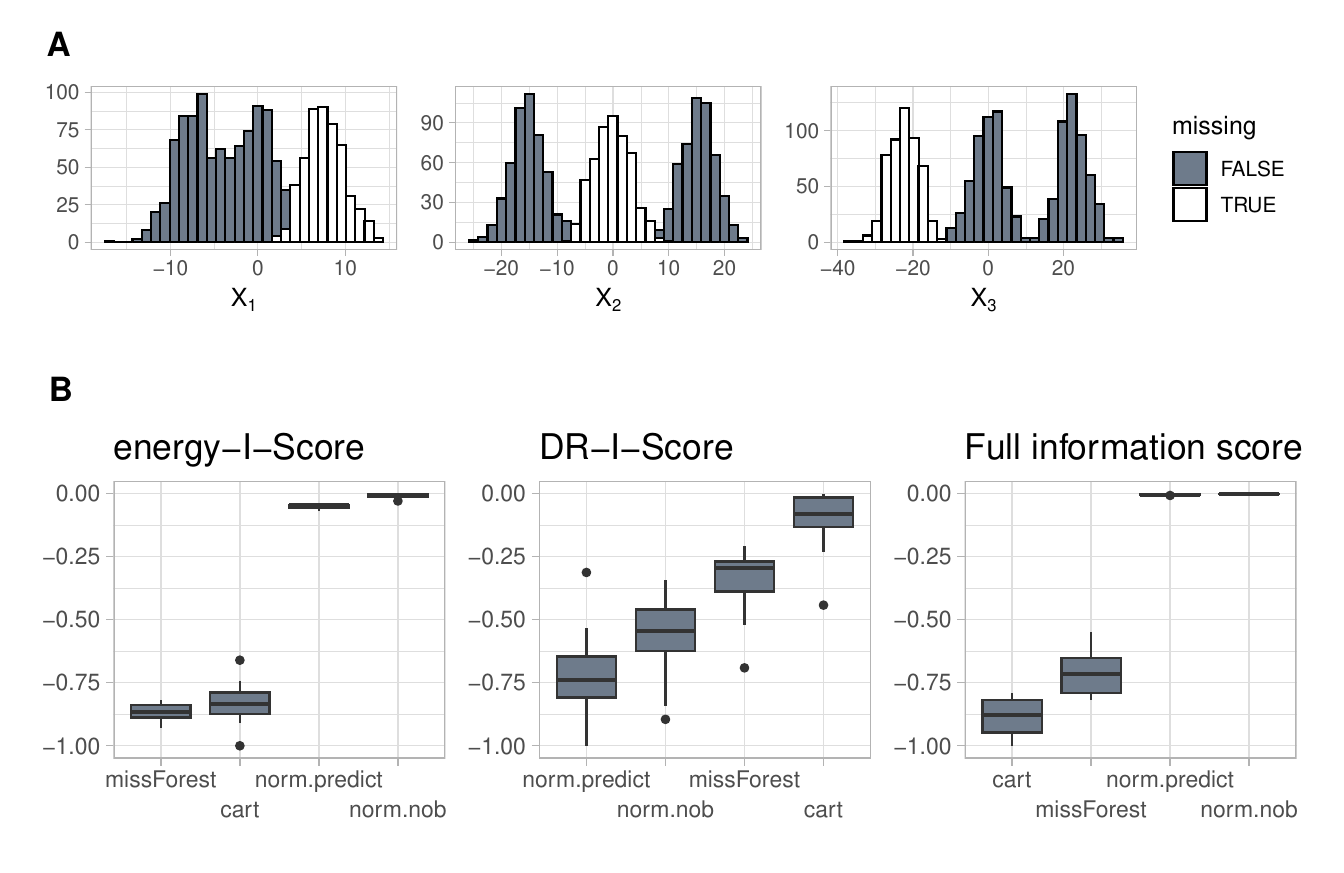} 
    \caption{Results for the Gaussian mixture model with distribution shift. (A) Illustration of $(X_1, X_2, X_3)$;  (B) Standardized scores for different imputations methods. Methods are ordered according to the mean score.}
    \label{fig:Application_2_Scores}
\end{figure}


We now turn to a Gaussian mixture model to put more emphasis on distribution shifts of observed variables. We take $d=6$ and 3 patterns,
\begin{align*}
    m_1 = (1, 0, 0, 0, 0 ,0), \ \   m_2 = (0, 1, 0, 0, 0 ,0),  \ \   m_3 = (0, 0, 1, 0, 0 ,0).
\end{align*}
The last three columns of fully observed variables, denoted $X_{O}^*$, are all drawn from three-dimensional Gaussians with means $(5,5,5)$, $(0,0,0)$ and $(-5,-5,-5)$ respectively, and a Toeplitz covariance matrix $\Sigma$ with $\Sigma_{i,j}=0.5^{|i-j|}$. Thus, there are relatively strong mean shifts between the different patterns. To preserve MAR, the (potentially unobserved) first three columns are built as 
\begin{align*}
    X_{O^c}^* = \mathbf{B}  X_{O}^* + \begin{pmatrix}
        \varepsilon_1\\
        \varepsilon_2\\
        \varepsilon_3
    \end{pmatrix} ,
\end{align*}
 where $\mathbf{B}$ is a $3 \times 3$ matrix of coefficients, $(\varepsilon_1, \varepsilon_2, \varepsilon_3)$ are independent $N(0,4)$ random errors and $O=\{4,5,6\}$ is the index of fully observed variables. For $\mathbf{B}$, we copy the vector $(0.5,1,1.5)$ three times, so that $\mathbf{B}$ has identical rows.
The data is thus Gaussian with linear relationships, but there is a strong distribution shift between the different patterns. However, this distributional shift only comes from the observed variables, leaving the conditional distributions of missing given observed unchanged. In particular, \Cref{ass_score} holds.

For each pattern, we generated 500 observations, resulting in $n=1.500$ observations and around $17\%$ missing values. In this example, we expect that the imputations able to adapt to shift in covariates will perform well, even if they are unable to capture complex dependencies between variables. Indeed, we note that $P^*$ corresponds to the Gaussian imputation (mice-norm.nob) with the (unknown) true parameters. As such, a proper score should rank mice-norm.nob highest. In contrast, forest-based imputations should perform the worst here, as they may not be able to deal properly with the distribution shift. 
On the other hand, they might still be deemed better than mice-norm.predict, which only imputes the regression prediction. The results for the (standardize) full information score, energy-I-Score and DR-I-Score are given in Figure \ref{fig:Application_2_Scores}. The full information and energy-I-Score behave as expected, with mice-norm.nob in the first place, and the forest-based methods last. Interestingly, both score mice-norm.predict second. The DR-I-Score instead ranks the forest-based methods highest, unable to recognize mice-norm.nob as the best method. As \ref{CIMAR} and in particular \ref{EMAR} hold in this example, this is likely due to the difficulty of the random forest classifier, used in the DR-I-Score implementation, to deal with covariate shifts. This shows a clear advantage of our new score, as it does not rely on any auxiliary method to estimate $H_{X_j \mid X_{O_j}}$, but instead directly generates and evaluates samples quality from the imputation method.

Thus, despite the challenging setting, the energy-I-Score still provides a sensible ordering. An interesting difference between the energy-I-Score and the full information score is that the energy-I-Score scores missForest lower than the full information score. However, this makes sense, as missForest is more severely punished when it creates $N$ imputations with very limited variation. In this sense, the energy-I-Score, without having access to the true data, might actually give a more accurate picture of the correct ordering.

\subsection{Mixture Model with Nonlinear Relationships}\label{Sec_Gaussmixmodelnonlinar}

We now turn to a more complex version of the model in Section \ref{Sec_Gaussmixmodel} to add nonlinear relationships to the distributional shifts. Using the same missingness pattern, and Gaussian variables $X_O$ we use a nonlinear function $f$ for the conditional distribution:
\begin{align}\label{nonlinearshiftexample}
    X_{O^c}^* =  f(X_{O}^*) + \begin{pmatrix}
        \varepsilon_1\\
        \varepsilon_2\\
        \varepsilon_3
    \end{pmatrix},
\end{align}
with
\begin{align*}
    f(x_1,x_2,x_3)=(x_3 \sin(x_1 x_2), x_2 \cdot \mathbf{1}\{x_2 > 0\}, \arctan(x_1) \arctan(x_2)).
\end{align*}
This introduces nonlinear relationships between the elements of $X_{O^c}^*$ and $X_O^*$, though the conditional distribution of $X_{O^c}^* \mid X_{O}^*$ is still Gaussian and the missingness mechanism is CIMAR. Moreover, \Cref{ass_score} is met here. For each pattern, we generate 500 observations, resulting in $n=1'500$ observations and around $17\%$ of missing values.

In this example, the ability to generalize is important, as is the ability to model nonlinear relationships. Consequently, this is a very difficult example and the ordering of the energy-I-Score and the full information score shown in Figure \ref{fig:Application_3_Scores} is quite different. In particular, they do not agree on the best two methods, though they both rank mice-cart high. This serves to illustrate that while at least the energy-I-Score should be able to identify the ``ideal'' imputation, there is no guarantee for what happens when all imputations perform badly. The disagreement of the scores should therefore be seen as more of a testament that none of the methods performs well than a sign that the scores themselves are flawed. 



\begin{figure}
    \centering
    \includegraphics[width=0.9\linewidth]{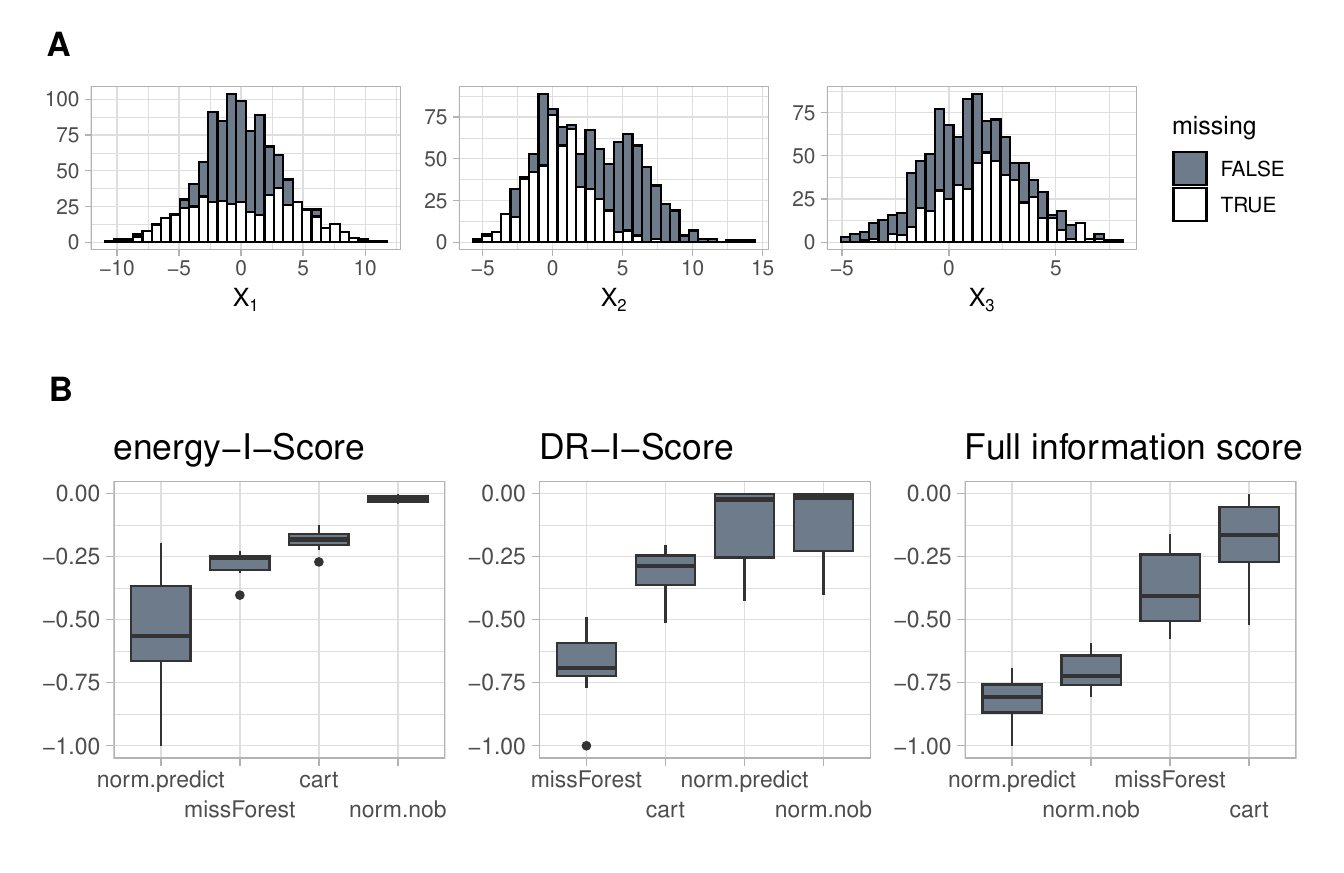} 
    \caption{Results for the Gaussian nonlinear mixture model. (A) Illustration of $(X_1, X_2, X_3)$, (B) Standardized scores for the Gaussian nonlinear mixture model. Methods are ordered according to the mean score.}
    \label{fig:Application_3_Scores}
\end{figure}

\section{Empirical Study: Downstream Tasks}\label{Sec_Empirical}

To further illustrate the use of the score, we consider an empirical analysis inspired by \citet{DML}, with added missing values: We study the effect of 401k eligibility on net worth using a double machine learning (DML) approach, based on a sample of 9,915 households from the 1990 Survey of Income and Program Participation (SIPP) in the U.S. The model we consider is the partially linear regression model, in the notation of \citet{DML}:
\begin{align*}
    Y = D\theta_0 + g_0(X) + U, \quad \mathbb{E}[U | X, D] = 0,
\\
D = m_0(X) + V, \quad \mathbb{E}[V | X] = 0
\end{align*}
The variable of interest $Y$ is the net financial assets of a household and following \cite{401KData, DML} we consider 9 covariates that include age, income, education, family size, marital status, and IRA participation. The treatment $D$ is taken to be 401(k) eligibility and it is commonly understood that unconfoundedness can be presumed when conditioned on the observed covariates (see e.g., \citep{401KData, DML} and the references therein).

Since the data set is complete, we introduce missing values artificially, using the \texttt{produce\_NA} function from the R-miss-tastic platform \citep{mayer2021rmisstastic}\footnote{The function is available at \url{https://rmisstastic.netlify.app/how-to/generate/misssimul}.}. This function is a wrapper around the widely used \texttt{ampute} routine from the \texttt{mice} package \citep{mice} that makes it easier to control the proportion of missing values. In addition, we also choose 4 random columns to be fully observed, ending up with a proportion of missingness of 6.7\%. We note that while it might be more desirable to consider a data set with real missing values, the use of artificial missing values allows for a comparison of the estimator under imputation to the one using full data. Moreover, the \texttt{ampute} function we use is an involved procedure to generate realistic missing values and is in no way tailored to our analysis.

After introducing missing values in the complete data set, we score six widely used imputation methods using the energy-I-Score; mice-rf, mice-cart, mice-sample, missForest, GAIN \citep{GAIN} and knn-imputation \citep{troyanskaya_missing_2001}. The latter two have not been used in Section \ref{Sec_Simulation} but are often the first choice in practice. We then impute the data set 10 times and estimate $\theta_0$ using the DML approach as in \cite{DML} for each imputation. As the first step of DML is to estimate $g_0$ and $m_0$ using a machine learning method, we use four different methods: LASSO (via cross-validated \texttt{glmnet} \citep{lasso}), Random Forest (\texttt{ranger} \citep{ranger}), Decision Trees (\texttt{rpart} \citep{rpart}) and Boosted Trees (\texttt{xgboost} \citep{xgboost}). In addition, we do the same for the complete data as a point of comparison.

Figure \ref{fig:empirical_study} shows the result of this procedure, both for the estimation of $\theta_0$ as well as its standard errors $\mathrm{SE}(\hat{\theta}_0)$ over 10 imputations. The plots are ordered according to the score values, showing mice-rf as the best imputation, followed by mice-cart, with values 1371.82 and 1542.96 respectively. For the point estimation in the top row, the value of the score aligns well with the estimate on the full data for all ML methods except rpart, where mice-cart appears to do a little better than mice-rf. Note that following Rubin's rules \citep[Chapter 9]{VANBUUREN2018}, the final estimator of $\theta_0$ would be the mean over the 10 imputations.\footnote{We note that our focus is on single imputation and none of the methods considered here are proper imputation methods, in that they neglect the uncertainty originating from the imputation method itself, which might lead to undercoverage of confidence intervals.} For the standard error estimation in the bottom row, mice-cart is arguably the best method. This likely mirrors the fact that, just as in the example of Section \ref{Sec_Gaussmixmodelnonlinar} there is again not one single ideal method in this example. Instead mice-cart and mice-rf are closely tied, each better capturing different aspects of the overall distribution. The energy-I-Score reflects this quite well here.



\begin{figure}
    \centering
    \includegraphics[width=1\linewidth]{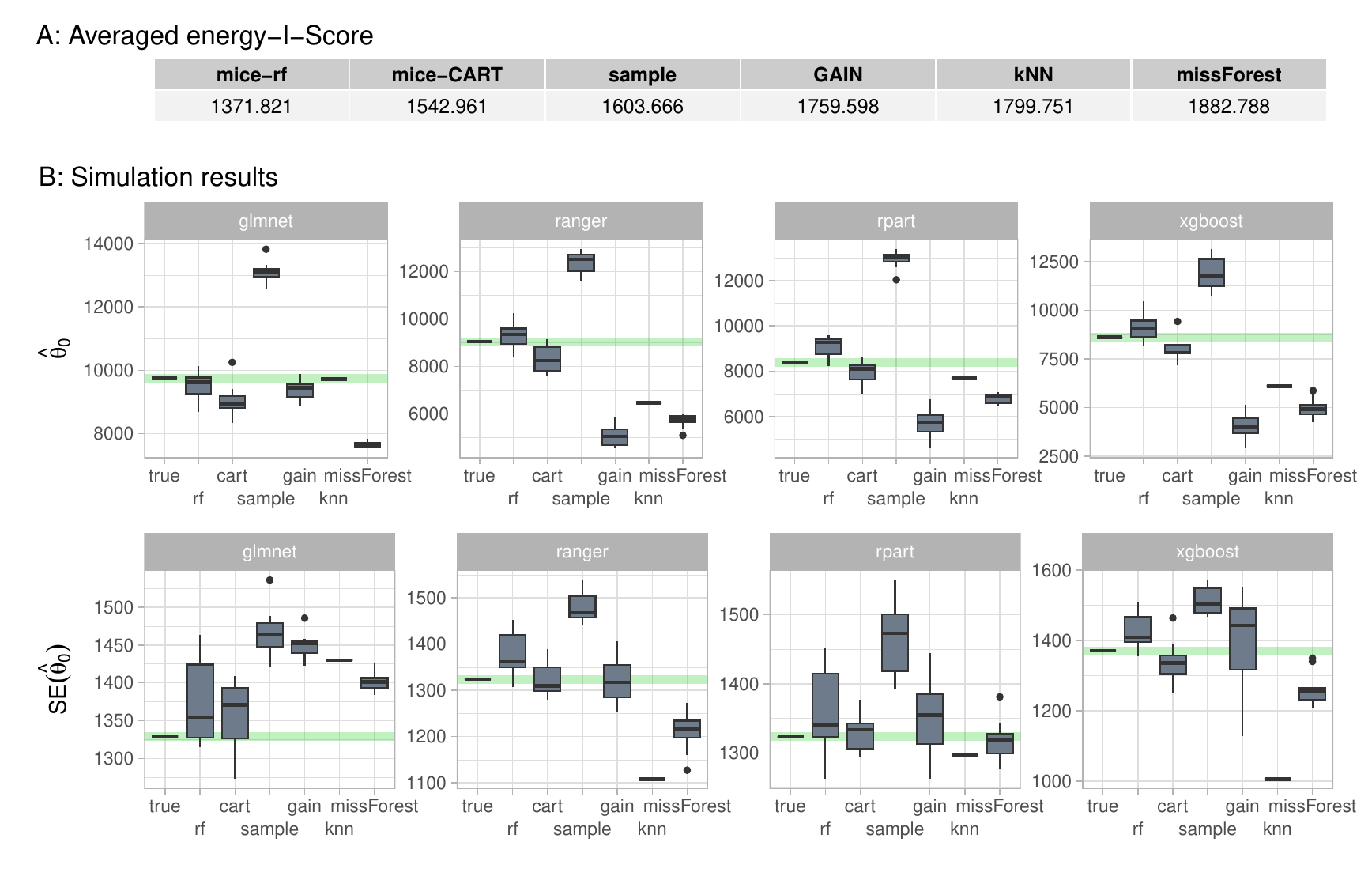} 
    \caption{(A) Energy-I-Score averaged across all replications. (B) Result of the DML approach based on different machine learning methods and different imputation methods, over 10 imputations, for the point estimation $\hat{\theta}_0$ (above) and standard error estimation $SE(\hat{\theta}_0)$ below. In each graph imputations are ordered according to the value of the energy-I-Score (best score to worst score from left to right). Green line denotes the result obtained on complete data.}
    \label{fig:empirical_study}
\end{figure}

\section{Conclusion}\label{Sec_Discussion}



In this paper, we developed a new I-Score to rank imputation methods. We studied the problem of scoring under \ref{PMMMAR} and defined a new missingness assumption, \Cref{ass_score}, which allowed us to construct a proper I-Score. We also studied an alternative score that utilizes a more straightforward adaptation of \ref{PMMMAR} in Appendix \ref{Sec_noprojectionsscore}.

As noted above, Assumption \ref{ass_score} is not the same as MAR and in fact, it is not weaker than \ref{EMAR}. However, adapting the score for \ref{EMAR} is straightforward. In fact, the adapted score for \ref{EMAR} and the presented score for Assumption \ref{ass_score} might be seen as the two extremes of a continuum: The score under \ref{EMAR} makes use of only the fully observed pattern, while the score under Assumption \ref{ass_score} considers the largest subset of fully observed variables such that we are able to consider all patterns. Consider a fixed $X_j$ and associated $O_j$. If we start to remove patterns, the number of observed variables in $O_j$ will increase. If $0 \in \mathcal{M}$, and we keep removing patterns $O_j$ will eventually contain all variables $X_{-j}$. In this case the only pattern left will be $m=0$, such that we obtain the score under \ref{EMAR}. 

While we focused on the version of the energy-I-Score using projections and variable-wise fully observed data, we also developed a score in Appendix \ref{Sec_noprojectionsscore} that is able to handle any missingness pattern. We denote this score by energy-I-Score$^*$. Although this score is not proper in examples such as Example \ref{Example1}, even under independence, the counterexamples used to show propriety does not hold use explicit pattern information and are often not available in practice. As such, the score has similar or better performance than the energy-I-Score in our empirical examples. On the other hand, the score is slow to implement and is less sample efficient than the energy-I-Score, making it much less practical. In particular, the test sets will tend to be much smaller for each pattern $m$ for the energy-I-Score$^*$, than for the energy-I-Score. As such, both scores have their advantages and downsides, leaving room for improvement.


\section*{Acknowledgements}

This work is part of the DIGPHAT project which was supported by a grant from the French government, managed by the National Research Agency (ANR), under the France 2030 program, with reference ANR-22-PESN-0017.
KG would like to acknowledge the support of the National Science Centre, Poland, under grant number 2021/43/O/ST6/02805.

\clearpage

\appendix

\section{Implementation Details}\label{Sec_Details}

In this section we detail the implementation of the energy-I-Score. First, in practice, $O_j$ might be empty for all $j$. In this case, we simply remove patterns until $O_j$ contains at least one element. This leads to the algorithm described in Algorithm \ref{algorithm1}.

\begin{algorithm}[H]\label{algorithm1}
\KwIn{Incomplete and imputed datasets, respectively $X, \tilde{X} \in \mathbb{R}^{n \times p}$, number of samples from imputation distribution $N \in \mathbb{N}$, imputation function $\mathcal{I}$.}
\KwOut{energy-I-Score estimate: $\widehat{S}_{\text{NA}}(H, P)$} 

\BlankLine
Define a set of variables with missing values: $\mathcal{S} = \{j: L_j^c \neq \emptyset\}$.

\ForEach{$j \in \mathcal{S}$}{
    $~~$\\
    \begin{enumerate}

        \item Skip current $j$ and return \texttt{NA} if the number of missing values $\left|\{i: m_{i, \cdot} \in L^c_j\}\right|$ or non-missing values $\big|\{i: m_{i, \cdot} \in L_j\}\big|$ in $j$ is less than $10$.

    \item  Get $O_j$ denoting the set of columns (excluding $j$) that are fully observed for non-missing rows of $X_j$. Namely, $O_j = \bigcap_{m \in L_j} \{l: m_l=0\}$.
    \begin{itemize}
            \item If $O_j = \emptyset$, then set $O_j = \{k^*\}$ where $k^*$ is the index of the variable (excluding $j$) that shares the largest number of jointly observed values with variable $j$, i.e., the one maximizing the number of observations not missing in both columns:
    
$$k^* = \argmax_{k \in \{1, \ldots, p\}\setminus\{j\}}  \big|\{i: m_{i, \cdot} \in L_j \cap L_k\}\big|.$$
    \end{itemize}

     \item Partition the dataset splitting it into train and test as follows:
$$
\text{Training set} = \begin{bmatrix}
\texttt{NA} & \tilde{X}_{\{i: m_{i, \cdot} \in L_j\}, O_j} \\
\tilde{X}_{\{i: m_{i, \cdot} \in L^c_j\}, j} & \tilde{X}_{\{i: m_{i, \cdot} \in L^c_j\}, O_j}
\end{bmatrix}, ~~ \text{Test set} = \begin{bmatrix}
\tilde{X}_{\{i: m_{i, \cdot} \in L_j\}, j}   \\
\end{bmatrix}. 
$$
        \item Draw samples $\tilde{X}_{i,j}^{(1)},  \ldots, \tilde{X}_{i,j}^{(N)}$ from  $H_{X_j|X_{O_j}, M_j=1}$ by imputing the \texttt{NA} part in the training set $N$ times via $\mathcal{I}$.

        \item Calculate     $$\widehat{S}^j_{\textrm{\tiny NA}}( H, P )= \frac{1}{|\{i: m_{i, \cdot} \in L_j\}|} \sum_{{i: m_{i, \cdot} \in L_j}}  \left(\frac{1}{2N^2} \sum_{l=1}^N \sum_{\ell=1}^N | \tilde{X}_{i,j}^{(l)} - \tilde{X}_{i,j}^{(\ell)}    | -  \frac{1}{N} \sum_{l=1}^N | \tilde{X}_{i,j}^{(l)} - x_{i,j} |   \right),$$ according to \eqref{newscore}, where $x_{i,j}$ is $i^{th}$ element in the test set.
    
        \item Compute weight $w_j = \dfrac{1}{n^2} \left|{\{i: m_{i, \cdot} \in L^c_j\}}\right| \cdot \big|\{i: m_{i, \cdot} \in L_j\}\big|$.
    \end{enumerate}
}
\Return{$\displaystyle \widehat{S}_{\textrm{\tiny NA}}(H, P) =\frac{1}{|\mathcal{S}|} \sum_{ j \in \mathcal{S}} w_j\widehat{S}_{\textrm{\tiny NA}}^j(H, P).$}
\end{algorithm}

Finally, when using the energy-I-Score a variable $X_j$ might be categorical. In such cases, the value of $\widehat{S}^j_{\textrm{\tiny NA}}(H, P)$ is computed based on a one-hot encoding of $X_j$. Specifically, we represent $X_j$ as a set of binary variables $X_{j_1}, \ldots, X_{j_p}$, corresponding to its $p$ possible categories, and treat them jointly as a single multivariate variable. Accordingly, instead of a univariate sample, we consider a multivariate sample $\tilde{\mathbf{X}}_{i,j}^{(1)}, \ldots, \tilde{\mathbf{X}}_{i,j}^{(N)}$, where each  
$\tilde{\mathbf{X}}_{i,j}^{(k)} = \left(\tilde{X}_{i,j_1}^{(k)}, \ldots, \tilde{X}_{i,j_p}^{(k)}\right)$ for  $k = 1, \ldots, N,$
represents one draw from the conditional distribution of the one-hot encoded representation of $X_j$ given $X_{O_j}$. The score is then computed as follows
$$\widehat{S}^j_{\textrm{\tiny NA}}( H, P )= \frac{1}{|\{i: m_{i, \cdot} \in L_j\}|} \sum_{{i: m_{i, \cdot} \in L_j}}  \left(\frac{1}{2N^2} \sum_{l=1}^N \sum_{\ell=1}^N \| \tilde{\mathbf{X}}_{i,j}^{(l)} - \tilde{\mathbf{X}}_{i,j}^{(\ell)}  \| -  \frac{1}{N} \sum_{l=1}^N \| \tilde{\mathbf{X}}_{i,j}^{(l)} - \mathbf{x}_{i,j} \|   \right).$$

\section{Further Empirical Results}\label{Sec_ExperimentsAppendix}

Here we consider further experiments. We start by analysing the uniform example under dependence, as well as a counter example, showing empirically that the energy-I-Score need not be strictly proper. Finally, we analyze the choice of $N$ in two examples.

\subsection{Uniform Example with Dependence}\label{Sec_Uniformwithdependence}

Here we consider the uniform example of Section \ref{Sec_PMMMARExample} but with dependence. In particular, we jointly draw $(Y_1, Y_2, Y_3)$ from a Gaussian with mean zero and covariance matrix $\Sigma$ with $\Sigma_{i,j}=0.7^{|i-j|}$. Then we transform $Y_j$ into uniform random variable by $X_j^*=\Phi(Y_j)$, where $\Phi$ is the cdf of the $N(0,1)$ distribution. This produces uniform random variables $X_1^*, X_2^*, X_3^*$ on $[0,1]$ with dependence and we introduce missingness in the same way as in Section \ref{Sec_PMMMARExample}. Scoring results are shown for $n=2000$ in Figure \ref{fig:uniform_dep} with the correct imputation distribution now denoted as ``dep\_runif''. Despite Assumption \ref{ass_score} not being true in this example, the ordering of the energy-I-Score is quite reasonable, showing dep\_runif first, closely followed by other methods that replicate the distribution well (according to the full information score).

\begin{figure}
    \centering
    \includegraphics[width=1\linewidth]{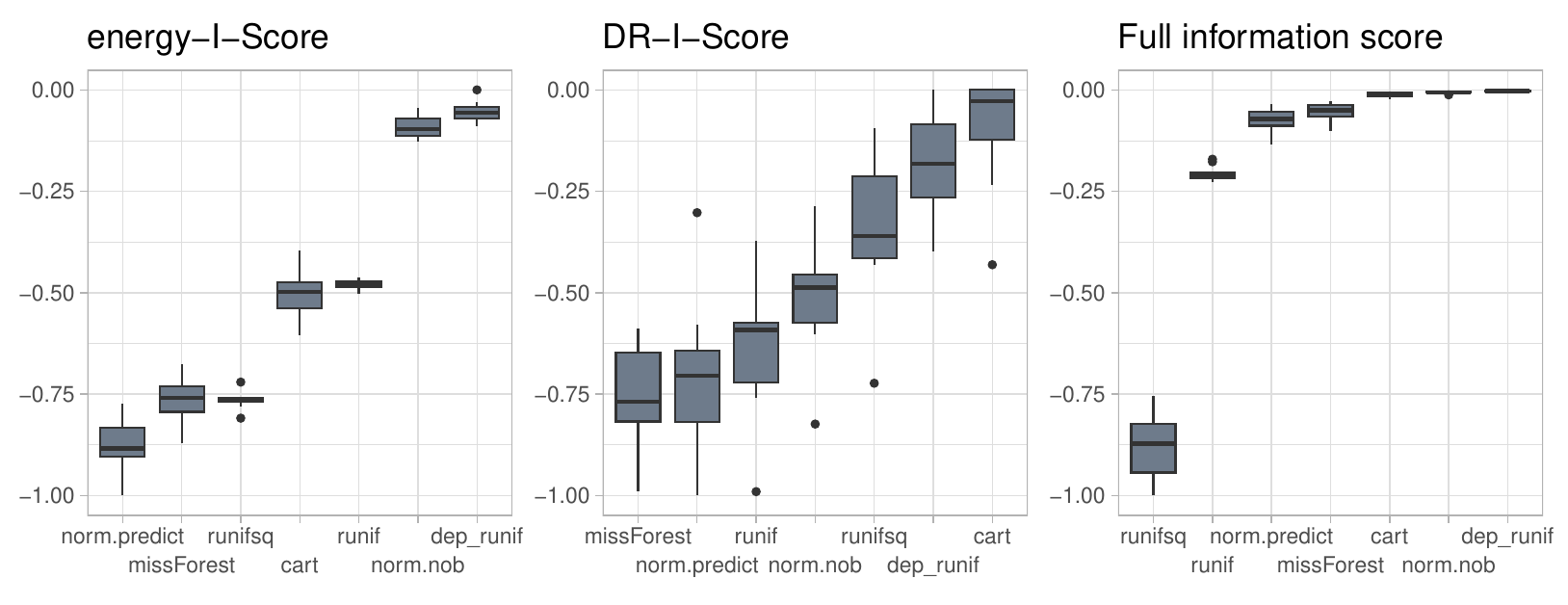} 
        \caption{Standardized scores for different imputations methods for the uniform example with dependence. Methods are ordered according to the mean score. }
    \label{fig:uniform_dep}
\end{figure}

\subsection{Strict Propriety Counter Example}\label{Sec_StrictProprietyCounterexample}

We now consider an example illustrating that the energy-I-Score is not strictly proper. In this example, we generate $2000$ observations from a $6$-dimensional normal distribution with means equal to $0$ and variances equal to $1$. In the random vector $X = (X_1, \ldots, X_6)$, all variables are independent except for the first two, with $\mathrm{cov}(X_1, X_2) = 0.7$. We consider three missingness patterns 

\begin{align*}
    m_1 = (1, 0, 0, 0, 0 ,0), \ \   m_2 = (0, 1, 0, 0, 0 ,0),   \ \   m_3 = (0, 0, 0, 0, 0 ,0).
\end{align*}

Each observation is randomly assigned one of the above patterns with equal probability, i.e., each pattern occurs with probability $1/3$. This results in approximately $11\%$ of missing values in the dataset. The missing values were imputed by sampling from either the conditional normal distribution (using mice with the norm.nob method) or an independent normal distribution (using rnorm). Naturally, in this case, sampling from the independent distribution is incorrect. Figure \ref{fig:nostrictproprety_2} shows the imputed variables $X_1$ and $X_2$ using both methods.

\begin{figure}
    \centering
    \includegraphics[width=1\linewidth]{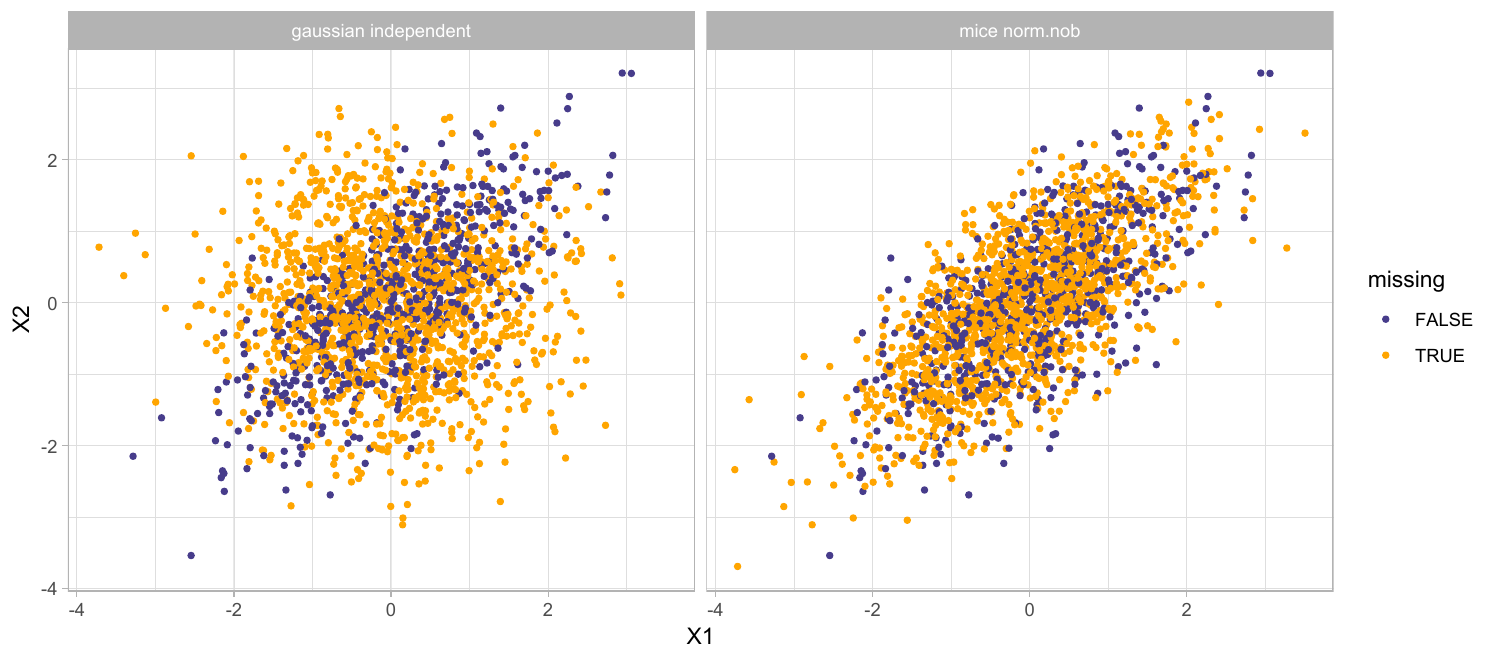} 
    \caption{Illustration of the imputation using Gaussian independent distribution (rnorm) and Gaussian conditional distribution (using mice norm.nob) in the Strict Propriety Counter Example.}
    \label{fig:nostrictproprety_2}
\end{figure}

The energy-I-Score is computed for $j \in \lbrace 1, 2 \rbrace$. In both cases, the set of observed variables is $O_j = \lbrace 3, 4, 5, 6 \rbrace$. As a result of the projection of the energy-I-Score, it never scores $(X_1,X_2)$ jointly and cannot differentiate the imputation method based on independent normal sampling and the correct imputation, with dependence between $X_1$ and $X_2$. Note that in this simple example, the DR-I-Score is strictly proper. However, depending on the projections chosen, this will not generally be the case in more complex examples.

\begin{figure}
    \centering
    \includegraphics[width=1\linewidth]{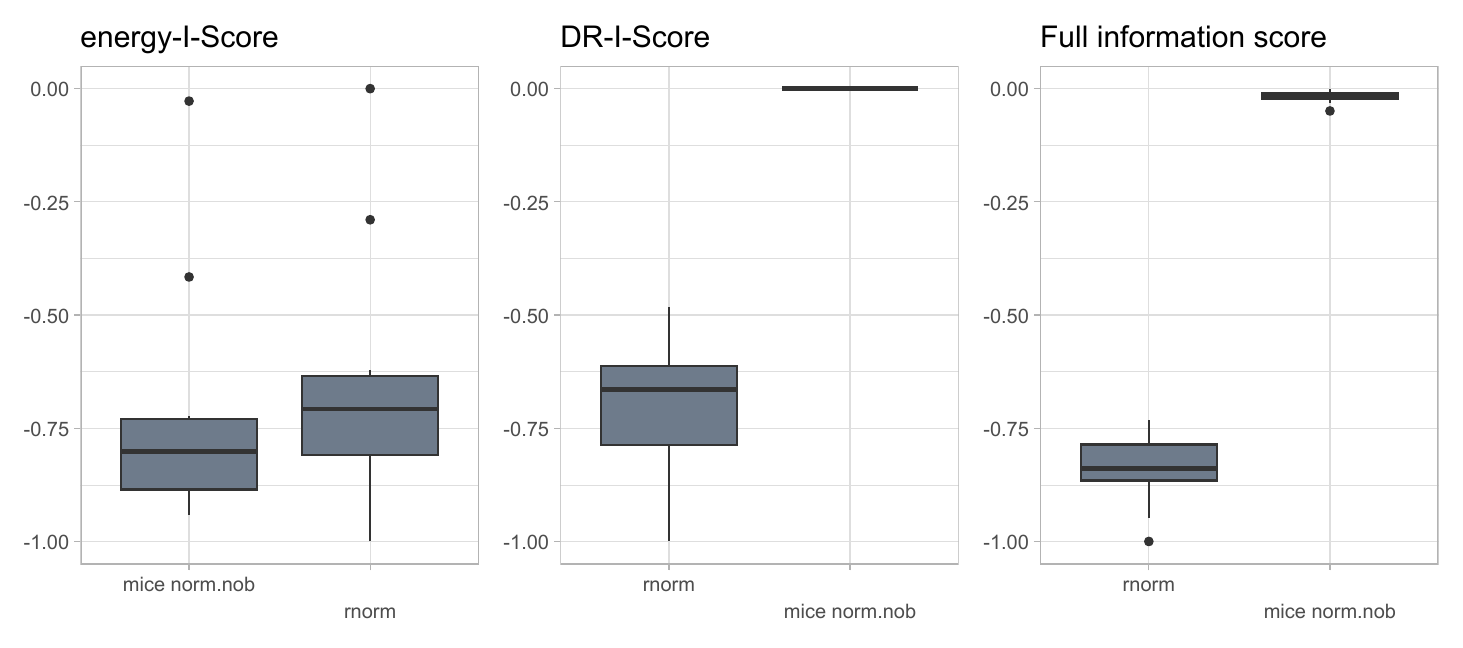} 
    \caption{Standardized scores for the Strict Propriety Counter Example. Methods are ordered according to the mean score.}
    \label{fig:nostrictproprety}
\end{figure}

\subsection{The Choice of $N$}\label{App_ChoiceofN}

To study the effect of $N$ in our examples, we rerun the uniform example of Section \ref{Sec_PMMMARExample} and the Gaussian mixture model from Section \ref{Sec_Gaussmixmodel} for $N$ starting at 5 and then ranging from 10 to 100 in steps of 10. Results are shown in Figure \ref{fig:N_choice_Uniform} for the uniform and in Figure \ref{fig:N_choice_Gaussian} for the Gaussian mixture example. Surprisingly, both appear to work quite well already for $N=5$, though particularly in the difficult uniform example, runif tends to be put in first place more reliably for $N=20$ and above.

\begin{figure}
    \centering
    \includegraphics[width=0.9\linewidth]{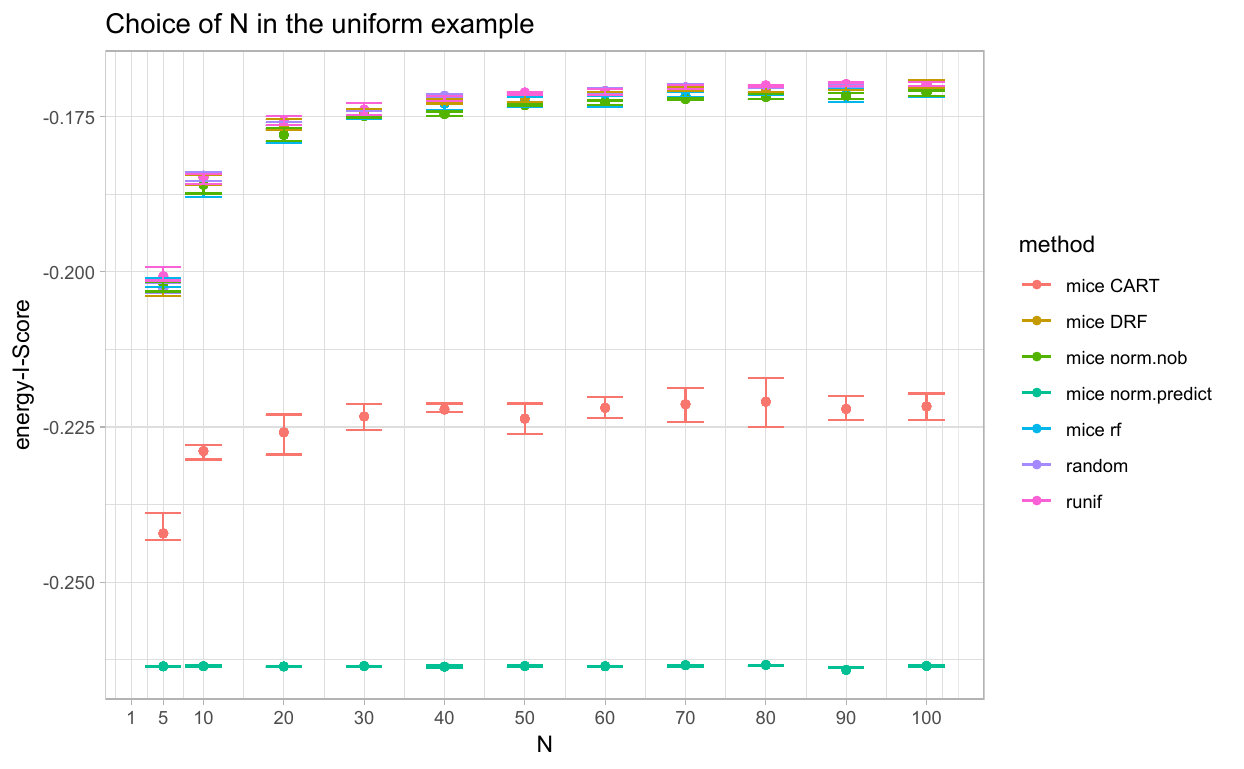} 
    \caption{Results from the uniform example in Section \ref{Sec_PMMMARExample} for different $N$ over 10 repetitions.}
    \label{fig:N_choice_Uniform}
\end{figure}

\begin{figure}
    \centering
    \includegraphics[width=0.9\linewidth]{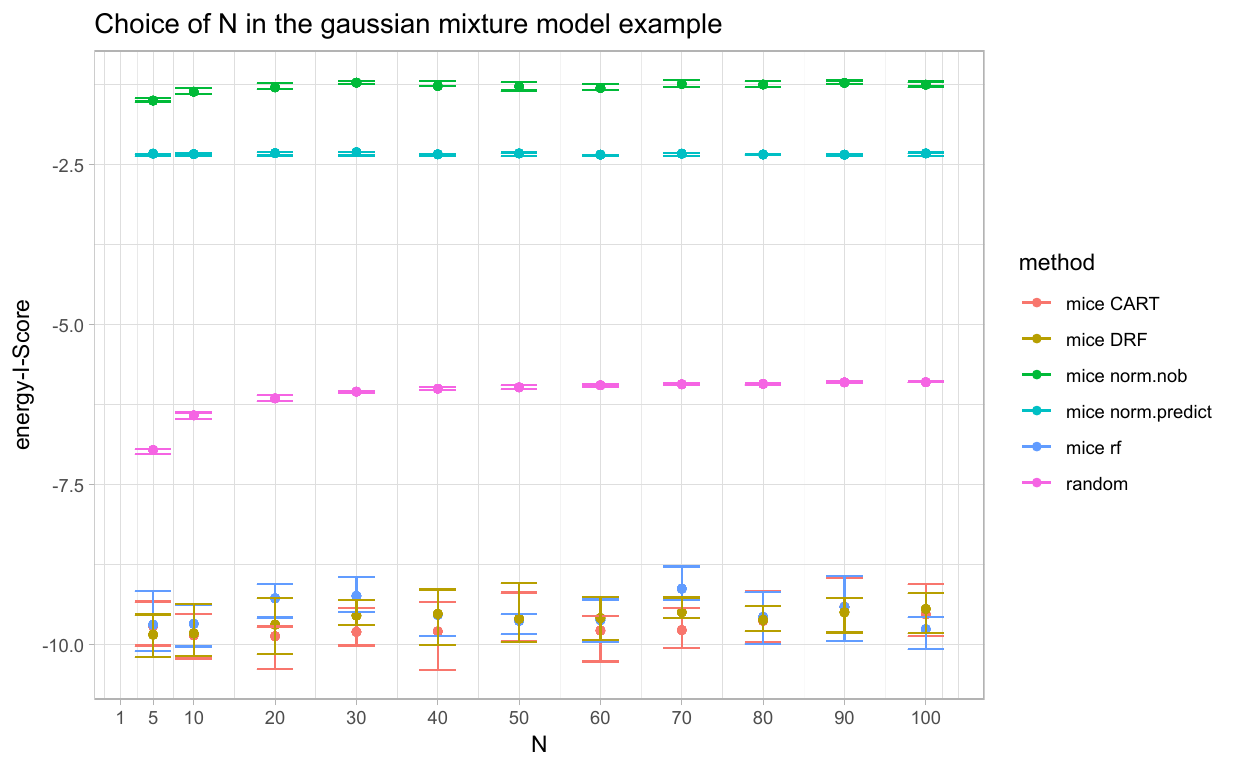} 
    \caption{Results from the Gaussian mixture example in Section \ref{Sec_Gaussmixmodelnonlinar} for different $N$ over 10 repetitions.}
    \label{fig:N_choice_Gaussian}
\end{figure}





\section{Proofs}\label{Sec_Proofs}

\identificationprop*

\begin{proof}
    \citet{näf2024goodimputationmarmissingness} showed that for all $x \in \X$ such that $p^*(x_{-j} \mid M_j=0) > 0$,
    \[
    p^*(x_j \mid x_{-j}, M_j=0 )=p^*(x_j \mid x_{-j}).
    \]
    Moreover, it follows from \ref{PMMMAR} directly that for all $x \in \X$ such that $p^*(x_{-j} \mid M_j=1) > 0$, 
     \[
    p^*(x_j \mid x_{-j}, M_j=1 )=p^*(x_j \mid x_{-j}).
    \]
    Thus for any $x \in \X$, $p^*(x_j \mid x_{-j}, M_j =1 )=p^*(x_j \mid x_{-j})$, showing the result.
\end{proof}

\proprietyprop*

\begin{proof}
We first note that for $j \in \mathcal{S}$, Condition \ref{ass_score_j} implies that $P^*_{X_j \mid X_{O_j}, M_j=1}=P^*_{X_j \mid X_{O_j}, M_j=0}=P^*_{X_j \mid X_{O_j}}$. Thus, considering the projection $(X_{A}, M_{A})$, with $A=O_j \cup \{j\}$, $P^*_{X_j \mid X_{O_j}}$ is the correct imputation. Indeed, for all $x \in \mathcal{X}$,
\begin{align*}
    p^*(x_j \mid x_{O_j})&=p^*(x_j \mid x_{O_j}, M_j=0) \Prob(M_j=0 \mid x_{O_j}) + p^*(x_j \mid x_{O_j}, M_j=1) \Prob(M_j=1 \mid x_{O_j})\\
    &=p^*(x_j \mid x_{O_j}) \Prob(M_j=0 \mid x_{O_j}) + p^*(x_j \mid x_{O_j}, M_j=1) \Prob(M_j=1 \mid x_{O_j}).
\end{align*}
This is equivalent to
\begin{align*}
    p^*(x_j \mid x_{O_j})\Prob(M_j=1 \mid x_{O_j})=p^*(x_j \mid x_{O_j}, M_j=1) \Prob(M_j=1 \mid x_{O_j}),
\end{align*}
showing that also $p^*(x_j \mid x_{O_j})=p^*(x_j \mid x_{O_j},M_j=1)$ as claimed. 
    We show that for each $j \in \mathcal{S}$, 
    \begin{align*}
     S_{NA}^{j}(H, P) \leq  S_{NA}^{j}(P^*, P)
    \end{align*}
    holds. To ease notation, we define
    \begin{align*}
        es(H,y) = \frac{1}{2} \E_{\substack{X \sim H\\X' \sim H}}[| X-X' |]- \E_{X \sim H}[ | X - y |]
    \end{align*}
    By (strict) propriety of the energy score (see e.g., \citet{gneiting})
    \begin{align*}
        \E_{Y \sim P^*_{X_j \mid x_{O_j}, M_j=0 }}[  es(H_{X_j \mid x_{O_j}, M_j=0 },Y)] \leq \E_{Y \sim P^*_{X_j \mid x_{O_j}, M_j=0 }}[  es(P^*_{X_j \mid x_{O_j}, M_j=0 },Y)].
    \end{align*}
    Taking expectations over $X_{O_j} \sim P^*_{X_{O_j} \mid M_j=0}$ on both sides shows that
    \begin{align}\label{inequality1}
           S_{NA}^{j}(H, P)= \E[ \E[  es(H_{X_j \mid X_{O_j}, M_j=0 },Y)]] \leq  \E[\E[  es(P^*_{X_j \mid X_{O_j}, M_j=0 },Y)]],
    \end{align}
    where we omitted the subscripts for a lighter notation. Moreover, \Cref{ass_score_j} implies that, for $X_{O_j} \sim P^*_{X_{O_j} \mid M_j=0}$, $P^*_{X_j \mid X_{O_j}, M_j=0 }=P^*_{X_j \mid X_{O_j}}=P^*_{X_j \mid X_{O_j}, M_j=1 }$ a.s., and
    \begin{align}\label{inequality2}
        \E[\E[  es(H^*_{X_j \mid X_{O_j}},Y)]]=\E[\E[  es(P^*_{X_j \mid X_{O_j}},Y)]]=S_{NA}^{j}(P^*, P).
    \end{align}
    Combining \eqref{inequality1} and \eqref{inequality2} gives the result for one $j$. Under Assumption \ref{ass_score}, we thus have that
\begin{align*}
    S_{\textrm{\tiny NA}}(H, P) = \frac{1}{|\mathcal{S}|} \sum_{ j \in \mathcal{S}} S_{\textrm{\tiny NA}}^j(H, P) \leq  S_{\textrm{\tiny NA}}(P^*, P) = \frac{1}{|\mathcal{S}|} \sum_{ j \in \mathcal{S}} S_{\textrm{\tiny NA}}^j(P^*, P),
\end{align*}
and thus the score is proper.
\end{proof}




\begin{lemma}\label{CoolLemma}
    Assumption \ref{ass_score} neither implies \ref{PMMMAR} nor is it implied by it.
\end{lemma}

To prove Lemma \ref{CoolLemma}, we present one example that is \ref{PMMMAR} but does not meet Assumption \ref{ass_score} and one example where Assumption \ref{ass_score} holds but \ref{PMMMAR} does not:

    \begin{example}[\ref{PMMMAR} does not imply Assumption \ref{ass_score}]\label{Example2}

Consider the same setting as in Example \ref{Example1},
\begin{align}
\mathcal{M}= \{ m_1, m_2, m_3 \}=\left\{\begin{pmatrix} 0 & 0 & 0\end{pmatrix}, \begin{pmatrix}0 & 1 & 0\end{pmatrix}, \begin{pmatrix} 1 & 0 & 0 \end{pmatrix} \right \},
\end{align}
and $X^*_1$, $X^*_2$, $X^*_3$ each following a uniform distributed on $[0,1]$ and, 
\begin{align*}
    \Prob(M=m_1 \mid x) &= \Prob(M=m_1 \mid x_1) = x_1/3\\
    \Prob(M=m_2 \mid x) &= \Prob(M=m_2 \mid x_1) = 2/3-x_1/3\\
    \Prob(M=m_3 \mid x) &= \Prob(M=m_3) = 1/3.
\end{align*}
However, now we assume a dependence between $X_1$ and $X_2$, for instance by using the approach in Section \ref{Sec_Uniformwithdependence}, such that $p^*(x_1 \mid x_2,x_3)=p^*(x_1 \mid x_2)$ for $x_1, x_2, x_3$ in a set with nonzero probability.  

In this case \ref{PMMMAR} still holds holds, but 
\begin{align*}
    \Prob(M_2=0 \mid x_2,x_3)=\int \Prob(M=m_2 \mid x_1) p^*(x_1 \mid x_2,x_3)=\Prob(M=m_2 \mid x_2),
\end{align*}
showing that $M_2 \indep X_2 \mid X_3$, and thus Assumption \ref{ass_score}, does not hold.
\end{example}

\begin{example}[Assumption \ref{ass_score} does not imply \ref{PMMMAR}]\label{Example3}

Consider
\begin{align*}
\mathcal{M}= \left\{\begin{pmatrix} 0 & 0 & 0\end{pmatrix}, \begin{pmatrix}1 & 0 & 0\end{pmatrix}, \begin{pmatrix} 0 & 1 & 0 \end{pmatrix}, \begin{pmatrix} 1 & 1 & 0 \end{pmatrix} \right \},
\end{align*}
and $(X^*_1, X^*_2, X^*_3)$ independently uniformly distributed on $[0,1]$. Moreover, now we assume that
\begin{align*}
    \Prob(M=m_1 \mid x)&= \Prob(M=m_3 \mid x) =(1-\mathbf{1}\{ x_2 > 0.5\} \cdot 0.8) \cdot 0.5 \\
    \Prob(M=m_2 \mid x)&= \Prob(M=m_4 \mid x)=\mathbf{1}\{ x_2 > 0.5\} \cdot 0.8 \cdot 0.5.
\end{align*}
Contrary to Example \ref{Example1}, \ref{PMMMAR} no longer holds here, as $\Prob(M=m_3 \mid x), \Prob(M=m_4 \mid x)$ depend on $x_2$ which is missing. 
Nonetheless, for $O_1=O_2 = \{3\}$, we have that
\begin{align*}
    \Prob(M_1=1 \mid x_1,x_3)&=\int (\Prob(M=m_2 \mid x_2) +\Prob(M=m_4 \mid x_2)) p^*(x_2 \mid x_1,x_3) dx_2\\
    &=\int \mathbf{1}\{ x_2 > 0.5\}*0.8 \  p^*(x_2 \mid x_1,x_3) dx_2 \\
    &= 0.5 \cdot 0.8.
\end{align*}
Thus, $X_1 \indep M_1 \mid X_3$ and Condition \ref{ass_score_j} holds for $j=1$. Similarly, 
\begin{align*}
    \Prob(M_2=1 \mid x_2,x_3)&=\int (\Prob(M=m_3 \mid x_2) +\Prob(M=m_4 \mid x_2)) p^*(x_2 \mid x_3) dx_2\\
    &=\int 0.5 \  p^*(x_2 \mid x_3) dx_2 \\
    &= 0.5,
\end{align*}
again showing $X_2 \indep M_2 \mid X_3$ and that Condition \ref{ass_score_j} holds for $j=2$.
\end{example}

\section{energy-I-Score without projections}\label{Sec_noprojectionsscore}



Here we formalize the score that was alluded to in the main text. The score does not use projections, although at the price of another missingness assumption that is strictly stronger than \ref{PMMMAR}. We start by introducing additional notation: As $X_{-j}$, $M_{-j}$ denotes all variables except variable $j$, we define $o(X,M)_{-j}=o(X_{-j}, M_{-j})$.

Following the intuition in Section \ref{Sec_Scoringdifficulty}, we first define the following population score:

\begin{align}\label{Scoreforj1}
  &S^{j,1}_{\textrm{\tiny NA}}( H, P ) = \nonumber\\
    &\E_{(o(X,M)_{-j}, M) \sim P_{o(X,M)_{-j}\mid M}\times \Prob_{M\mid M \in L_j}}\Big[   \E_{\substack{X \sim H_{X_j \mid o(X,M)_{-j}, M }\\Y \sim P^*_{X_j \mid o(X,M)_{-j}, M}}}[ \| X - Y \|_{2}] \nonumber\\
  &- \frac{1}{2} \E_{\substack{ X \sim H_{X_j \mid o(X,M)_{-j}, M }\\ X' \sim H_{X_j \mid o(X,M)_{-j}, M}}}[\| X-X' \|_{2}] \Big]
\end{align}
where the outer expectation is taken over the joint distribution of $(o(X,M)_{-j}, M)$, for $M \in L_j$. Then, the full score is given as
\begin{align*}
    S_{\textrm{\tiny NA}}^1(H, P) = \frac{1}{|\mathcal{S}|} \sum_{ j \in \mathcal{S}} S_{\textrm{\tiny NA}}^{j,1}(H, P) ,
\end{align*}
where $\mathcal{S}\subset \{1,\ldots, d\}$. As discussed in Section \ref{Sec_Scoringdifficulty}, for $m \in L_j^c$ we would like to score $P^*_{X_j \mid o(X,m)_{-j}}$ highest, when in fact, for $m \in L_j$, $P^*_{X_j \mid o(X,m)_{-j}, M=m}$ is the correct imputation distribution. For instance, in Example \ref{Example1}, when scoring $X_1$ we would like to give highest score to the imputation that simply draws from independent uniform distributions, $p^*(x_1 \mid x_2, x_3)=\mathbf{1}\{x_1 \in [0,1]\}$ instead of $p^*(x_1 \mid x_2, x_3, M=m_1)=2x_1 \mathbf{1}\{x_1 \in [0,1]\}$ and $p^*(x_1 \mid x_2, x_3, M=m_2)=2/3 (2-x_1)\mathbf{1}\{x_1 \in [0,1]\}$. As such the following condition is necessary:

\begin{assumption}\label{ass_score_noproj_j0} 
        For all $j \in \mathcal{S}$,
    \begin{align}
&p^*( x_j \mid o(x ,m )_{-j}, M =m ) \nonumber \\
&= p^*(x_j\mid o(x ,m )_{-j})\ \  \forall m \in L_j, x \in \X_{\mid m}.
\end{align}
\end{assumption}

Since
\begin{align*}
    p^*(x_j \mid o(x,m)_{-j}, M=m )&=\int p^*(x_j, o^c(x,m) \mid o(x,m)_{-j}, M=m ) d o^c(x,m)\\
    &= \int p^*(x_j, o^c(x,m) \mid o(x,m)_{-j} ) d o^c(x,m)\\
    &=p^*(x_j \mid o(x,m)_{-j}),
\end{align*}

Assumption \ref{ass_score_noproj_j0} is implied by the following stronger assumption:

\begin{assumption} \label{ass_score_noproj_j1}
        For all $j \in \mathcal{S}$,
    \begin{align}
&p^*(o^c(x ,m ), x_j \mid o(x ,m )_{-j}, M =m ) \nonumber \\
&= p^*(o^c(x ,m ), x_j\mid o(x ,m )_{-j})\ \  \forall m \in L_j, x \in \X_{\mid m}.  
\end{align}
\end{assumption}




We note the difference of Assumption \ref{ass_score_noproj_j1} to \ref{PMMMAR}, which requires that for all $m \in \mathcal{M}$, $x \in \X_{\mid m}$,
    \begin{align*}
&p^*(o^c(x ,m ) \mid o(x ,m ), M =m )= p^*(o^c(x ,m )\mid o(x ,m )).   
\end{align*}
Thus we effectively ask that for $m \in L_j$, the relation of \ref{PMMMAR} holds as if $x_j$ was missing. This reveals two possible implementations of the score in \eqref{Scoreforj1}: First, an implementation without projection where for each $m \in L_j$, $(X_j, o^c(X,m)) \mid o(X,m)_{-j}$ is imputed jointly, while only $X_j$ is scored, as discussed in Section \ref{Sec_Scoringdifficulty}. Second, a more direct implementation, where for each $m \in L_j$, we impute $X_j \mid o(X,m)_{-j}$ instead. We consider the second approach below.

\begin{proposition}\label{proprietyofnewscore}
    Under \ref{PMMMAR}, $S_{\textrm{\tiny NA}}^1$ is a proper I-Score if and only if Assumption \ref{ass_score_noproj_j0} holds.
\end{proposition}

\begin{proof}
Take $j \in \mathcal{S}$ arbitrary. We first note that by propriety of the energy score, $S^{j,1}_{\textrm{\tiny NA}}( H, P )$ is maximal for the imputation $H^*$ that has $H^*_{X_j \mid o^c(X,m)_{-j}, M=m}=P^*_{X_j \mid o^c(X,m)_{-j}, M=m}$, for all $m \in L_j$. By Assumption \ref{ass_score_noproj_j0}, we have $H^*_{X_j \mid o^c(X,m)_{-j}, M=m}=P^*_{X_j \mid o^c(X,m)_{-j}}$ for all $m \in L_j$, showing propriety of $S_{\textrm{\tiny NA}}^1$.

Assume now that for a $j \in \mathcal{S}$, Assumption \ref{ass_score_noproj_j0} does not hold. Thus, there exists $m \in L_j$, such that $P^*_{X_j \mid o^c(X,m)_{-j}, M=m} \neq P^*_{X_j, \mid o^c(X,m)_{-j}}$ and $H^*$ reaches a higher score by the strict propriety of the energy distance used in $S^{j,1}_{\textrm{\tiny NA}}$.
\end{proof}

If Assumption \ref{ass_score_noproj_j0} holds, $S_{\textrm{\tiny NA}}^1$ is proper. The set $\mathcal{S}$ of $X_j$ that are scored can in principle be $\{1,\ldots,d\}$, in which case $S_{\textrm{\tiny NA}}^1$ would be a \emph{strictly proper score}. However, the larger $\mathcal{S}$, the stronger Assumption \ref{ass_score_noproj_j0} becomes. On the other hand, if the condition does not hold, there are examples for which $S_{\textrm{\tiny NA}}^1$ is not proper, as was already discussed in Section \ref{Sec_Scoringdifficulty}. For completeness, we consider again the uniform example both under dependence (used in Section \ref{Sec_Uniformwithdependence}) and independence (Example \ref{Example1}). As, 
    \begin{align*}
\mathcal{M}= \{ m_1, m_2, m_3 \}=\left\{\begin{pmatrix} 0 & 0 & 0\end{pmatrix}, \begin{pmatrix}0 & 1 & 0\end{pmatrix}, \begin{pmatrix} 1 & 0 & 0 \end{pmatrix} \right \},
\end{align*}
with
    \begin{align*}
    \Prob(M=m_1 \mid x) &= \Prob(M=m_1 \mid x_1) = x_1/3\\
    \Prob(M=m_2 \mid x) &= \Prob(M=m_2 \mid x_1) = 2/3-x_1/3\\
    \Prob(M=m_3 \mid x) &= \Prob(M=m_3) = 1/3.
\end{align*}
In this case, for $m_2 \in L_1$
$$
p^*(x_1 \mid x_3, M=m_2)=p^*(x_1 \mid x_3)\frac{2- x_1}{\int (2- x_1) p^*(x_1 \mid x_3)dx_1},
$$
and thus Assumption \ref{ass_score_noproj_j0} does not hold for $j=1$, both under independence and independence. This shows that the score might not be proper, even if the energy-I-Score with projections is. However, we note that in these counterexamples an imputation that uses pattern information, namely, $P^*_{o^c(X,M), X_j \mid o(X,M), M}$, has to be used to achieve higher scoring. If only imputation methods are compared, that do not consider additional information about patterns, such as FCS imputation, the score might still reliably find the best imputation method in practice. 

We now consider a practical estimation of the score: 
\begin{itemize}
    \item[1.] Randomly choose a test set, $\mathcal{T} \subset \{i: m_{i, \cdot} \in L_j\}$.
    \item[2.] Impute the training set $\mathcal{T}^c$.
    \item[3.] Randomly draw a pattern $M$ from the patterns in the test set ($M \in L_j$). Let $\mathcal{T}_{M} \subset \mathcal{T}$, $\mathcal{T}_{M}=\{i\in \mathcal{T}: m_{i,\cdot}=M\}$.
        \item[4.] Create a new data set by concatenating the observed $(x_{i,j}, (x_{i,\cdot},M)_{-j})$, $i \in \mathcal{T}_{M}$, and the imputed $(x_{i,j}, (x_{i,\cdot},M)_{-j})$, $i \in \mathcal{T}^c$, as in Figure \ref{fig:scoreillustrationXM}, and set the \emph{observed} observations of $X_{j}$ to missing, i.e. $x_{i,j}=\texttt{NA}$ for $i \in \mathcal{T}_{M}$.
    \item[5.] Approximate the sampling from $H_{X_j \mid o(X,M)_{-j}, M}$, by imputing the new dataset, $N$ times, leading to samples $X_{i,j}^{(l)}$, $l=1, \ldots, N$, for all $x_{i,j}$, $i \in \mathcal{T}_{M}$.
    \item[6.] Calculate 
    \[
    S^M_{j}( H, P )=\frac{1}{|\mathcal{T}_{M}|} \sum_{i \in \mathcal{T}_{M}}  \left(\frac{1}{2N^2} \sum_{l=1}^N \sum_{\ell=1}^N | \tilde{X}_{i,j}^{(l)} - \tilde{X}_{i,j}^{(\ell)}    | -  \frac{1}{N} \sum_{l=1}^N | \tilde{X}_{i,j}^{(l)} - x_{i,j} |   \right).
    \]
\end{itemize}
Repeating steps 3. -- 6. several times, we then calculate 
\begin{align}\label{newscore_noprojections}
    &\widehat{S}^{j,1}_{\textrm{\tiny NA}}( H, P )= \E_{M \sim \Prob_M}[S^M_{j}( H, P ) ].
\end{align}
Finally
\begin{align*}
    \widehat{S}_{\textrm{\tiny NA}}^1(H, P) = \frac{1}{|\mathcal{S}|} \sum_{ j \in \mathcal{S}} \widehat{S}_{\textrm{\tiny NA}}^{j,1}(H, P).
\end{align*}
We refer to $\widehat{S}_{\textrm{\tiny NA}}^1$ as energy-I-Score*.

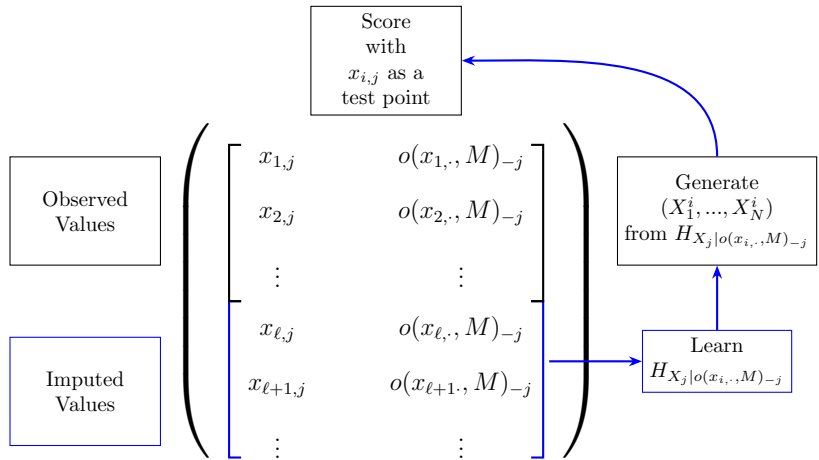
\begin{figure}
    \centering

    \begin{tikzpicture}[scale=0.8, every node/.style={scale=0.8},
    box/.style={draw, rectangle, minimum width=2.5cm, minimum height=1.8cm, align=center},
    smallbox/.style={draw, rectangle, minimum width=2cm, minimum height=1cm, align=center},
    bluebox/.style={draw=blue, rectangle, minimum width=2cm, minimum height=1cm, align=center},
    arrow/.style={-{Stealth[length=5pt]}, blue, thick}
]

\node[box] (observed) at (-1,3) {Observed\\Values};
\node[box, draw=blue] (imputed) at (-1,0) {Imputed\\Values};

\matrix (data) [matrix of math nodes, row sep=0.2cm, column sep=0.8cm, nodes={minimum width=1.5cm}, font=\large] at (4,1.5) {
  x_{1,j} & o(x_{1, \cdot}, M)_{-j} \\
  x_{2,j} & o(x_{2,\cdot}, M)_{-j} \\
  \vdots & \vdots \\
  x_{\ell,j} & o(x_{\ell,\cdot}, M)_{-j} \\
  x_{\ell+1,j} & o(x_{\ell+1\cdot}, M)_{-j} \\
  \vdots & \vdots \\
};

\node at ($(data.north west)+(-0.5,-2.7)$) {\scalebox{2.2}{$\left(\vphantom{\begin{matrix}x_{1,j}\\x_{2,j}\\\vdots\\x_{\ell+1,j}\\\vdots\end{matrix}}\right.$}};
\node at ($(data.north east)+(0.5,-2.7)$) {\scalebox{2.2}{$\left.\vphantom{\begin{matrix}x_{1,j}\\x_{2,j}\\\vdots\\x_{\ell+1,j}\\\vdots\end{matrix}}\right)$}};

\node[box] (score) at (4,5.5) {Score\\with\\$x_{i,j}$ as a\\test point};
\node[box] (generate) at (9.5,3) {Generate\\$(X_1^{i},...,X_N^{i})$\\from $H_{X_j \mid o(x_{i,\cdot}, M)_{-j}}$};
\node[bluebox] (learn) at (9.5,0.5) {Learn\\$H_{X_j \mid o(x_{i,\cdot}, M)_{-j}}$};

\draw[arrow] ($(data.east)+(0,-1)$) to[out=0,in=180] (learn);
\draw[arrow] (learn) -- (generate);
\draw[arrow] (generate) to[out=90,in=0] (score);

\draw[blue, thick] ($(data.west)-(-0.3,0)$)-- ++(-0.2,0) -- ++(0,-2.6) -- ++(0.2,0);
\draw[blue, thick] ($(data.east)+(-0.3,0)$)-- ++(0.2,0) -- ++(0,-2.6) -- ++(-0.2,0);

\draw[black, thick] ($(data.west)-(-0.3,-2.6)$)-- ++(-0.2,0) -- ++(0,-2.6) -- ++(0.2,0);
\draw[black, thick] ($(data.east)+(-0.3,2.6)$)-- ++(0.2,0) -- ++(0,-2.6) -- ++(-0.2,0);

\end{tikzpicture}

    \caption{Conceptual illustration of the score approximation for a random draw of a pattern $M \in L_j$, similarly to Figure \ref{fig:scoreillustrationOj}. Contrary to the energy-I-Score with projections we now have a training set in blue (including observations $m_{i, \cdot} \in L_j$ and $m_{i, \cdot} \in L_j^c$) that was fully imputed and a test set in black of pattern $M \in L_j$.}
    \label{fig:scoreillustrationXM}
\end{figure}

In the following, we repeat all previous examples, with the new score added. Remarkably, the score has similar performance as the energy-I-Score with projections, while also being strictly proper in the example of Section \ref{Sec_StrictProprietyCounterexample}. 

\begin{figure}
    \centering
    \includegraphics[width=1\linewidth]{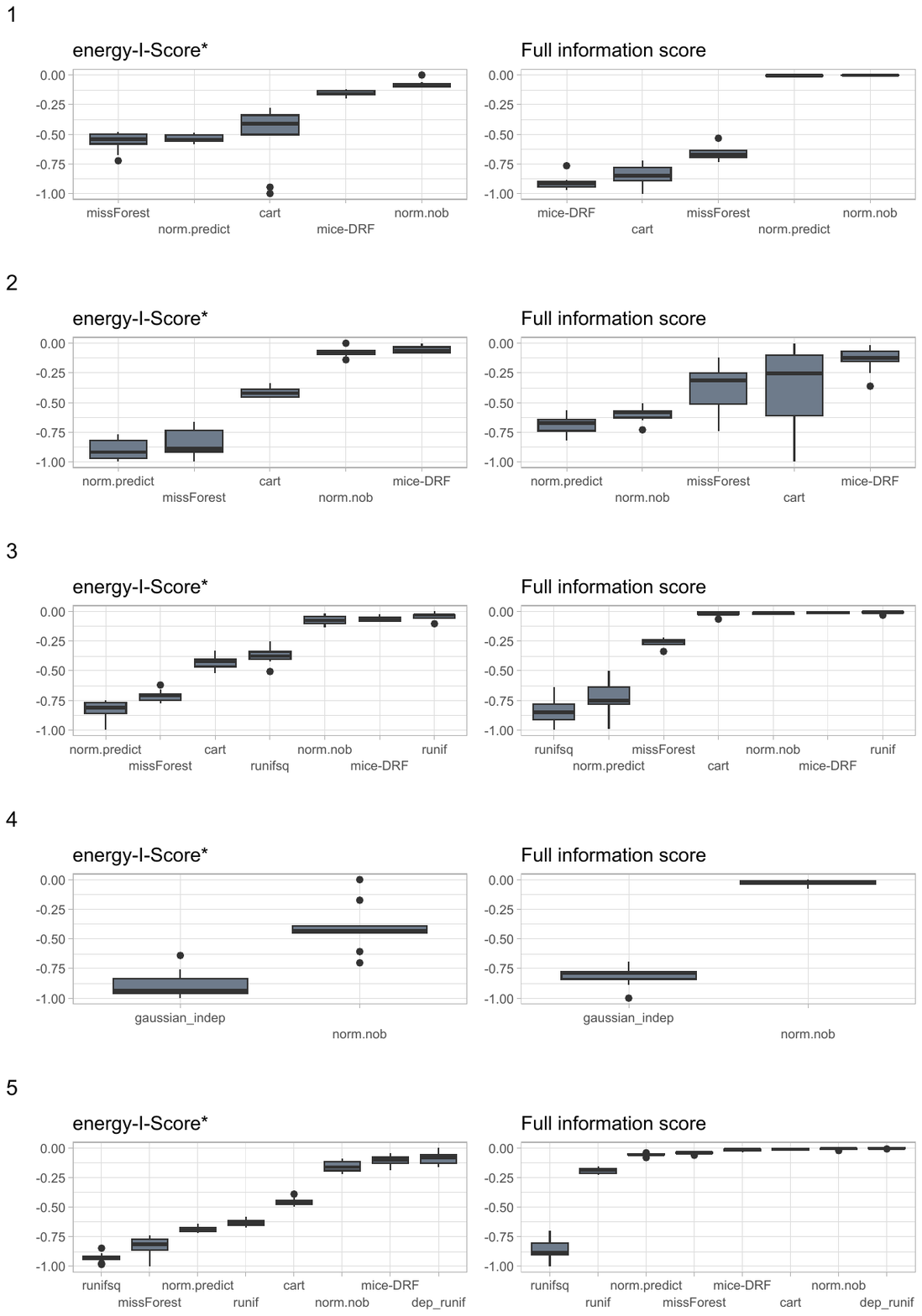} 
    \caption{(1) Gaussian Mixture Example of Section \ref{Sec_Gaussmixmodel}, (2) Nonlinear Mixture Example of Section \ref{Sec_Gaussmixmodelnonlinar}, (3) Independent Uniform Example of Section \ref{Sec_PMMMARExample}, (4) Non-Strict Propriety Example of Section \ref{Sec_StrictProprietyCounterexample} and (5) Dependent Uniform Example of Section \ref{Sec_Uniformwithdependence} with the new score.}
    \label{fig:NonlinearMixNewscore}
\end{figure}





\clearpage

\bibliographystyle{apalike}
\bibliography{biblio}

\end{document}